\documentclass[11pt,a4paper]{article} 
\usepackage[margin=1in]{geometry}
\usepackage{algorithmic}
\usepackage{algorithm}
\usepackage{times} 
\usepackage{amsfonts}
\usepackage{amsmath} 
\usepackage{amsthm} 
\usepackage{latexsym} 
\usepackage{epsfig}
\usepackage{latexsym} 
\usepackage{verbatim}
\usepackage{xcolor}
\usepackage{hyperref}
\hypersetup{colorlinks=true,citecolor=blue,linkcolor=blue,urlcolor=blue}
\usepackage{cleveref}
\usepackage{todonotes}

\begin{document}

\newtheorem{theorem}{Theorem}
\newtheorem*{theorem*}{Theorem}
\newtheorem{corollary}[theorem]{Corollary}
\newtheorem{prop}[theorem]{Proposition} 
\newtheorem{problem}[theorem]{Problem}
\newtheorem{lemma}[theorem]{Lemma} 
\newtheorem*{lemma*}{Lemma}
\newtheorem{remark}[theorem]{Remark}
\newtheorem{observation}[theorem]{Observation}
\newtheorem{defin}{Definition} 
\newtheorem{example}{Example}
\newtheorem{conj}{Conjecture} 
\newcommand{\PR}{\noindent {\bf Proof:\ }} 
\def\EPR{\hfill $\Box$\linebreak\vskip.5mm} 
 
\def\Pol{{\sf Pol}} 
\def\mPol{{\sf MPol}} 
\def\Polo{{\sf Pol}_1} 
\def\PPol{{\sf pPol\;}} 
\def\Inv{{\sf Inv}}
\def\mInv{{\sf MInv}} 
\def\Clo{{\sf Clo}\;} 
\def\Con{{\sf Con}} 
\def\concom{{\sf Concom}\;} 
\def\End{{\sf End}\;}
\def\Sub{{\sf Sub}\;} 
\def\Im{{\sf Im}} 
\def\Ker{{\sf Ker}\;} 
\def\H{{\sf H}}
\def\S{{\sf S}} 
\def\D{{\sf P}} 
\def\I{{\sf I}} 
\def\Var{{\sf var}} 
\def\PVar{{\sf pvar}} 
\def\fin#1{{#1}_{\rm fin}}
\def\P{{\sf P}} 
\def\Pfin{{\sf P_{\rm fin}}} 
\def\Id{{\sf Id}}
\def\R{{\rm R}} 
\def\F{{\rm F}} 
\def\Term{{\sf Term}}
\def\var#1{{\sf var}(#1)} 
\def\Sg#1{{\sf Sg}(#1)} 
\def\Sgg#1#2{{\sf Sg}_{#1}(#2)} 
\def\Cg#1{{\sf Cg}(#1)}
\def\tol{{\sf tol}} 
\def\rbcomp#1{{\sf rbcomp}(#1)}
\def\core{\mathsf{core}}
  
\let\cd=\cdot 
\let\eq=\equiv 
\let\op=\oplus 
\let\ot=\otimes
\let\omn=\ominus
\let\meet=\wedge 
\let\join=\vee 
\let\tm=\times
\def\ldiv{\mathbin{\backslash}} 
\def\rdiv{\mathbin/}
  
\def\typ{{\sf typ}} 
\def\zz{{\un 0}} 
\def\zo{{\un 1}}
\def\one{{\bf1}} 
\def\two{{\bf2}} 
\def\three{{\bf3}}
\def\four{{\bf4}} 
\def\five{{\bf5}}
\def\pq#1{(\zz_{#1},\mu_{#1})}
  
\let\wh=\widehat 
\def\ox{\ov x} 
\def\oy{\ov y} 
\def\oz{\ov z}
\def\of{\ov f} 
\def\oa{\ov a} 
\def\ob{\ov b} 
\def\oc{\ov c}
\def\od{\ov d} 
\def\oob{\ov{\ov b}} 
\def\rx{{\rm x}}
\def\rf{{\rm f}} 
\def\rrm{{\rm m}} 
\let\un=\underline
\let\ov=\overline 
\let\cc=\circ 
\let\rb=\diamond 
\def\ta{{\tilde a}} 
\def\tz{{\tilde z}}
  
  
\def\zZ{{\mathbb Z}} 
\def\B{{\mathcal B}} 
\def\P{{\mathcal P}}
\def\zL{{\mathbb L}} 
\def\zD{{\mathbb D}}
 \def\zE{{\mathbb E}}
\def\zG{{\mathbb G}} 
\def\zA{{\mathbb A}} 
\def\zB{{\mathbb B}}
\def\zC{{\mathbb C}} 
\def\zM{{\mathbb M}} 
\def\zR{{\mathbb R}}
\def\zS{{\mathbb S}} 
\def\zT{{\mathbb T}} 
\def\zN{{\mathbb N}}
\def\zQ{{\mathbb Q}} 
\def\zW{{\mathbb W}} 
\def\bK{{\bf K}}
\def\C{{\bf C}} 
\def\M{{\bf M}} 
\def\E{{\bf E}} 
\def\N{{\bf N}}
\def\O{{\bf O}} 
\def\bN{{\bf N}} 
\def\bX{{\bf X}} 
\def\GF{{\rm GF}} 
\def\cC{{\mathcal C}} 
\def\cA{{\mathcal A}}
\def\cB{{\mathcal B}} 
\def\cD{{\mathcal D}} 
\def\cE{{\mathcal E}} 
\def\cF{{\mathcal F}} 
\def\cG{{\mathcal G}} 
\def\cH{{\mathcal H}}
\def\cI{{\mathcal I}} 
\def\cL{{\mathcal L}} 
\def\cM{{\mathcal M}} 
\def\cP{{\mathcal P}} 
\def\cR{{\mathcal R}} 
\def\cRY{{\mathcal RY}}
\def\cS{{\mathcal S}} 
\def\cT{{\mathcal T}} 
\def\cX{{\mathcal X}} 
\def\cY{{\mathcal Y}} 
\def\oB{{\ov B}}
\def\oC{{\ov C}} 
\def\ooB{{\ov{\ov B}}} 
\def\ozB{{\ov{\zB}}}
\def\ozD{{\ov{\zD}}} 
\def\ozG{{\ov{\zG}}}
\def\tcA{{\widetilde\cA}} 
\def\tcC{{\widetilde\cC}}
\def\tcF{{\widetilde\cF}} 
\def\tcI{{\widetilde\cI}}
\def\tB{{\widetilde B}} 
\def\tC{{\widetilde C}}
\def\tD{{\widetilde D}} 
\def\ttB{{\widetilde{\widetilde B}}}
\def\ttC{{\widetilde{\widetilde C}}}
\def\tba{{\tilde\ba}} 
\def\ttba{{\tilde{\tilde\ba}}}
\def\tbb{{\tilde\bb}} 
\def\ttbb{{\tilde{\tilde\bb}}}
\def\tbc{{\tilde\bc}} 
\def\tbd{{\tilde\bd}}
\def\tbe{{\tilde\be}} 
\def\tbt{{\tilde\bt}}
\def\tbu{{\tilde\bu}} 
\def\tbv{{\tilde\bv}}
\def\tbw{{\tilde\bw}} 
\def\tdl{{\tilde\dl}} 
\def\ocP{{\ov\cP}}
\def\tzA{{\widetilde\zA}} 
\def\tzC{{\widetilde\zC}}
\def\new{{\mbox{\footnotesize new}}}
\def\old{{\mbox{\footnotesize old}}}
\def\prev{{\mbox{\footnotesize prev}}}
\def\oo{{\mbox{\sf\footnotesize o}}}
\def\pp{{\mbox{\sf\footnotesize p}}}
\def\nn{{\mbox{\sf\footnotesize n}}} 
\def\oR{{\ov R}}
  
  
\def\gA{{\mathfrak A}} 
\def\gV{{\mathfrak V}} 
\def\gS{{\mathfrak S}} 
\def\gK{{\mathfrak K}} 
\def\gH{{\mathfrak H}}
  
\def\ba{{\bf a}} 
\def\bb{{\bf b}} 
\def\bc{{\bf c}} 
\def\bd{{\bf d}} 
\def\be{{\bf e}} 
\def\bbf{{\bf f}} 
\def\bg{{\bf g}}
\def\bh{{\bf h}}
\def\bi{{\bf i}} 
\def\bm{{\bf m}} 
\def\bo{{\bf o}} 
\def\bp{{\bf p}} 
\def\bs{{\bf s}} 
\def\bu{{\bf u}} 
\def\bt{{\bf t}} 
\def\bv{{\bf v}} 
\def\bx{{\bf x}}
\def\by{{\bf y}} 
\def\bw{{\bf w}} 
\def\bz{{\bf z}}
\def\ga{{\mathfrak a}} 
\def\oal{{\ov\al}} 
\def\obeta{{\ov\beta}}
\def\ogm{{\ov\gm}} 
\def\oep{{\ov\varepsilon}}
\def\oeta{{\ov\eta}} 
\def\oth{{\ov\th}} 
\def\ovm{{\ov\mu}}
\def\ozero{{\ov0}}
  
  
\def\CCSP{\hbox{\rm c-CSP}} 
\def\CSP{{\rm CSP}} 
\def\NCSP{{\rm \#CSP}} 
\def\mCSP{{\rm MCSP}} 
\def\FP{{\rm FP}} 
\def\PTIME{{\bf PTIME}} 
\def\GS{\hbox{($*$)}} 
\def\ry{\hbox{\rm r+y}}
\def\rb{\hbox{\rm r+b}} 
\def\Gr#1{{\mathrm{Gr}(#1)}}
\def\Grp#1{{\mathrm{Gr'}(#1)}} 
\def\Grpr#1{{\mathrm{Gr''}(#1)}}
\def\Scc#1{{\mathrm{Scc}(#1)}} 
\def\rel{R} 
\def\relo{Q}
\def\rela{S} 
\def\dep{\mathsf{dep}}
\def\Filt{\mathrm{Ft}}
\def\Filts{\mathrm{Fts}} 
\def\Agr{$\mathbb{A}$}
\def\Al{\mathrm{Alg}}
\def\Sig{\mathrm{Sig}}
\def\strat{\mathsf{strat}}
\def\relmax{\mathsf{relmax}}
\def\srelmax{\mathsf{srelmax}}
\def\Meet{\mathsf{Meet}}
\def\amax{\mathsf{amax}}
\def\as{\mathsf{as}}
\def\star{\hbox{$(*)$}}
\def\bmal{{\mathbf m}}
\def\Af{\mathsf{Af}}
\let\sqq=\sqsubseteq

  
\let\sse=\subseteq 
\def\ang#1{\langle #1 \rangle}
\def\angg#1{\left\langle #1 \right\rangle}
\def\dang#1{\ang{\ang{#1}}} 
\def\vc#1#2{#1 _1\zd #1 _{#2}}
\def\zd{,\ldots,} 
\let\bks=\backslash 
\def\red#1{\vrule height7pt depth3pt width.4pt
\lower3pt\hbox{$\scriptstyle #1$}}
\def\fac#1{/\lower2pt\hbox{$\scriptstyle #1$}}
\def\me{\stackrel{\mu}{\eq}} 
\def\nme{\stackrel{\mu}{\not\eq}}
\def\eqc#1{\stackrel{#1}{\eq}} 
\def\cl#1#2{\arraycolsep0pt
\left(\begin{array}{c} #1\\ #2 \end{array}\right)}
\def\cll#1#2#3{\arraycolsep0pt \left(\begin{array}{c} #1\\ #2\\
#3 \end{array}\right)} 
\def\clll#1#2#3#4{\arraycolsep0pt
\left(\begin{array}{c} #1\\ #2\\ #3\\ #4 \end{array}\right)}
\def\cllll#1#2#3#4#5#6{ \left(\begin{array}{c} #1\\ #2\\ #3\\
#4\\ #5\\ #6 \end{array}\right)} 
\def\pr{{\rm pr}}
\let\upr=\uparrow 
\def\ua#1{\hskip-1.7mm\uparrow^{#1}}
\def\sua#1{\hskip-0.2mm\scriptsize\uparrow^{#1}} 
\def\lcm{{\rm lcm}} 
\def\perm#1#2#3{\left(\begin{array}{ccc} 1&2&3\\ #1&#2&#3
\end{array}\right)} 
\def\w{$\wedge$} 
\let\ex=\exists
\def\NS{{\sc (No-G-Set)}} 
\def\lev{{\sf lev}}
\let\rle=\sqsubseteq 
\def\ryle{\le_{ry}} 
\def\ryprec{\le_{ry}}
\def\os{\mbox{[}} 
\def\zs{\mbox{]}}
\def\link{{\sf link}}
\def\solv{\stackrel{s}{\sim}} 
\def\mal{\mathbf{m}}
\def\precs{\prec_{as}}

  
\def\lb{$\linebreak$}  
  
\def\ar{\hbox{ar}} 
\def\Im{{\sf Im}\;} 
\def\deg{{\sf deg}}
\def\id{{\rm id}}
  
\let\al=\alpha 
\let\gm=\gamma 
\let\dl=\delta 
\let\ve=\varepsilon
\let\ld=\lambda 
\let\om=\omega 
\let\vf=\varphi 
\let\vr=\varrho
\let\th=\theta 
\let\sg=\sigma 
\let\Gm=\Gamma 
\let\Dl=\Delta
  
  
\font\tengoth=eufm10 scaled 1200 
\font\sixgoth=eufm6
\def\goth{\fam12} 
\textfont12=\tengoth 
\scriptfont12=\sixgoth
\scriptscriptfont12=\sixgoth 
\font\tenbur=msbm10
\font\eightbur=msbm8 
\def\bur{\fam13} 
\textfont11=\tenbur
\scriptfont11=\eightbur 
\scriptscriptfont11=\eightbur
\font\twelvebur=msbm10 scaled 1200 
\textfont13=\twelvebur
\scriptfont13=\tenbur 
\scriptscriptfont13=\eightbur
\mathchardef\nat="0B4E 
\mathchardef\eps="0D3F

\newcommand{\infin}[1]{}

\newcommand{\Mnote}[1]{\todo[size=footnotesize,bordercolor=black,color=blue!40]{#1}}
\newcommand{\snote}[1]{\textcolor{blue!40}{(Standa: #1)}}
\newcommand{\anote}[1]{\textcolor{red}{(Andrei: #1)}}

\def\bZ{{\bf Z}}
\newcommand{\ignore}[1]{}

\title{Satisfiability of commutative vs. non-commutative CSPs\thanks{An extended abstract of part of this work appeared in Proceedings of ICALP\,(A) 2025. This work was supported by UKRI EP/X024431/1 and NSERC Discovery grant. For the purpose of Open Access, the authors have applied a CC BY public copyright licence to any Author Accepted Manuscript version arising from this submission. All data is provided in full in the results section of this paper.}}
\author{Andrei A.\ Bulatov\\SFU \and Stanislav \v{Z}ivn\'y\\University of Oxford} 
\date{} 
\maketitle

\begin{abstract}
The Mermin-Peres magic square is a celebrated example of a system of Boolean
  linear equations that is not (classically) satisfiable but is satisfiable via
  linear operators on a Hilbert space of dimension four. A natural question is
  then, for what kind of problems such a phenomenon occurs? Atserias, Kolaitis,
  and Severini answered this question for all Boolean Constraint Satisfaction
  Problems (CSPs): For \textsc{0-Valid-SAT}, \textsc{1-Valid-SAT},
  \textsc{2-SAT}, \textsc{Horn-SAT}, and \textsc{Dual Horn-SAT}, classical
  satisfiability and operator satisfiability is the same  and thus there is no
  gap; for all other Boolean CSPs, these notions differ as there are gaps, i.e.,
  there are unsatisfiable instances that are  satisfiable via operators on Hilbert spaces. 

We generalize their result to CSPs on \emph{arbitrary finite} domains and give
  an almost complete classification: First, we show that NP-hard CSPs admit a
  separation between classical satisfiability and satisfiability via operators
  on finite- and infinite-dimensional Hilbert spaces. Second, we show that
  tractable CSPs of bounded width have no satisfiability gaps of any kind.
  Finally, we show that tractable CSPs of unbounded width can simulate, in a
  satisfiability-gap-preserving fashion, linear equations over an Abelian group
  of prime order $p$; for such CSPs, we obtain a separation of classical
  satisfiability and satisfiability via operators on infinite-dimensional
  Hilbert spaces. Furthermore, if $p=2$, such CSPs also have gaps separating
  classical satisfiability and satisfiability via operators on finite- and infinite-dimensional Hilbert spaces. 
\end{abstract}

\section{Introduction}

Symmetry leads to efficient computation. This phenomenon has
manifested itself in several research areas that have one aspect in common, namely a
model of computation with local constraints that restrict the solution space of
the problem of interest. An elegant way to describe such problems is in the
framework of Constraint Satisfaction Problems (CSPs).
CSPs have driven some of the most influential
developments in theoretical computer science, from NP-completeness to the PCP
theorem to semidefinite programming algorithms to the Unique Games Conjecture.
The mathematical structure of tractable decision
CSPs~\cite{Bulatov17:focs,Zhuk20:jacm}, infinite-domain CSPs~\cite{Bodirsky2021complexity,Bodirsky23:jacm},  optimization CSPs~\cite{tz16:jacm}, as well as approximable
CSPs~\cite{Raghavendra08:everycsp,Brown-CohenRaghavendra15}, is now known to be linked to certain
forms of higher-order symmetries of the solution spaces.
A recently emerging research direction links 
CSPs with foundational
topics in physics and quantum
computation~\cite{Cleve14:icalp,Cleve17:jmp-perfect,AKS19:jcss,Paddock25:ijm,Mancinska20:focs}.

\paragraph{Constraint Satisfaction Problems}
CSPs capture some of the most fundamental
computational problems, including graph and hypergraph colorings, linear
equations, and variants and generalizations of satisfiability. Informally, one
is given a set of variables and a set of constraints, each depending only on
constantly many variables. Given a CSP instance, the goal is to find an assignment
of values to all the variables so that all constraints are satisfied. For
example, if the domain is $\{r,g,b\}$ and the constraints are of the form
$R(x,y)$, where 
\[R=\{(r,g),(g,r),(g,b),(b,g),(r,b),(b,r)\}\]
is the binary disequality relation on $\{r,g,b\}$, we obtain
the classic graph \textsc{3-Colorability}
problem. If the domain is $\{r,g,b,o\}$ and the
constraints are of the form $R(x,y)$, where
\[R=\{(r,g),(g,r),(g,b),(b,g),(b,o),(o,b),(o,r),(r,o)\},\]
we obtain a variant of
the graph \textsc{4-Colorability} problem in which adjacent vertices must be
assigned different colors and, additionally, red and blue vertices must not be adjacent
and green and orange vertices must not be adjacent. 

Back in 1978, Schaefer famously classified all Boolean CSPs as solvable in
polynomial time or NP-hard~\cite{Schaefer78:stoc}.\footnote{We call a CSP
\emph{Boolean} if the domain of all variables is of size two. Some papers call
such CSPs binary. We use the term \emph{binary} for a relation of arity two and
for CSPs whose constraints involve binary relations.} The
tractable cases are the standard textbook problems, namely 
\textsc{2-SAT}, \textsc{Horn-SAT}, \textsc{dual
Horn-SAT}, and system of linear equations on a two-element set.\footnote{There
are also two trivial, uninteresting cases called 0-valid and 1-valid.} 
Hell and Ne\v{s}et\v{r}il studied a special case of CSPs known
as graph homomorphisms~\cite{hell2004graphs}. These are CSPs in which all constraints
involve the same binary symmetric relation, i.e., a graph. The above-mentioned 
\textsc{3-Colorability} problem is the homomorphism problem to $K_3$,
the undirected clique on three vertices, say $\{r,g,b\}$. The above-mentioned
variant of \textsc{4-Colorability} is the homomorphism problem to $C_4$, the
undirected cycle on four vertices, say $\{r,g,b,o\}$. Generalizing greatly
the classic result of Karp that \textsc{$k$-Colorability} is solvable in polynomial time for $k\leq 2$ and NP-hard
for $k\geq 3$ and other concrete problems such as the (tractable) variant of
\textsc{4-Colorability} above, Hell and Ne\v{s}et\v{r}il obtained in
1990 a complete classification of such CSPs~\cite{HellN90}.
Motivated by these two results and the quest to identify the largest subclass of
NP that could exhibit a dichotomy and thus avoid NP-intermediate cases, Feder
and Vardi famously conjectured that \emph{all} CSPs on finite domains admit a
dichotomy; i.e., are either solvable in polynomial time or are
NP-hard~\cite{Feder98:monotone}.

Following the so-called algebraic approach to CSPs, pioneered by Jeavons, Cohen
and Gyssens~\cite{Jeavons97:closure} and Bulatov, Jeavons, and
Krokhin~\cite{Bulatov05:classifying}, the conjecture was resolved in the
affirmative in 2017 by Bulatov~\cite{Bulatov17:focs} and, independently,
by Zhuk~\cite{Zhuk17:focs,Zhuk20:jacm}. The algebraic approach allows for a very clean and
precise characterization of what makes certain CSPs computationally tractable
--- this is captured by the notion of polymorphisms, which can be thought of as
multivariate symmetries of solutions spaces of CSPs.
Along the way to the resolution of the Feder-Vardi dichotomy conjecture, the
algebraic approach has been successfully used to
establish other results about CSPs, 
e.g., characterizing the power of local consistency
algorithms for CSPs~\cite{Bulatov09:width,Barto14:local}, characterizing robustly solvable
CSPs~\cite{Guruswami12:tght,Dalmau13:toct,Barto16:sicomp}, classifying the
complexity of exact optimization CSPs~\cite{tz16:jacm,Kolmogorov17:sicomp}, the
tremendous progress on classifying the complexity of CSPs on infinite
domains~\cite{Bodirsky2021complexity}, and very recently using SDPs for robustly
solving certain promise CSPs~\cite{BGS23:stoc}.

\paragraph{Operator Constraint Satisfaction Problems}

Consider the following instance of a Boolean CSP, consisting in nine variables $x_1,\ldots,x_9$ over the Boolean domain $\{-1,+1\}$ and the following six
constraints:

\begin{equation}\label{eq:magic}
  \begin{split}
    x_1 x_2 x_3 = +1, \qquad x_1 x_4 x_7 = +1,\\
    x_4 x_5 x_6 = +1, \qquad x_2 x_5 x_8 = +1,\\
    x_7 x_8 x_9 = +1, \qquad x_3 x_6 x_9 = -1.
  \end{split}
\end{equation}

Graphically, this system of equations can be represented by a square, where each
equation on the left of~(\ref{eq:magic}) comes from a row, and each equation on the right
of~(\ref{eq:magic}) comes from a column.

\begin{center}
{\large
\begin{tabular}{|c|c|c|}
 \hline
  $x_1$ & $x_2$ & $x_3$ \\\hline 
  $x_4$ & $x_5$ & $x_6$ \\\hline 
  $x_7$ & $x_8$ & $x_9$ \\\hline
\end{tabular}
}
\begin{tabular}{c}
    $+1$\\
    $+1$\\
    $+1$
\end{tabular}
\\[-0.7em]
\begin{tabular}{ccc}
  $+1$ & $+1$ & $-1$
\end{tabular}
\phantom{\begin{tabular}{c} $+1$\\ $+1$\\ $+1$ \end{tabular}}
\end{center}
The system of equations~(\ref{eq:magic}) has no solution in the Boolean domain
$\{-1,+1\}$: By multiplying the left-hand sides of all equations we get $+1$
because every variable occurs twice in the system and $x_i^2=+1$ for every
$1\leq i\leq 9$. However, by multiplying the right-hand sides of all equations,
we get $-1$. Note that this argument used implicitly the assumption that the
variables commute pairwise even if they do not appear in the same equation,
which is true over $\{-1,+1\}$.
Thus, this argument does not hold if one assumes that 
only variables occurring
in the same equation commute pairwise. In fact, Mermin famously established that
the system~(\ref{eq:magic}) has a solution consisting of linear operators on a
Hilbert space of dimension four~\cite{Mermin1990simple,Mermin1993hidden} and the
construction is now know as the Mermin-Peres magic
square~\cite{Peres1990incompatible}. This construction proves the
Bell-Kochen-Specker theorem on the impossibility to explain quantum
mechanics via hidden variables~\cite{Bell1966problem,Kochen67}. 

Any Boolean CSP instance, just like the one above, can be associated with a certain non-local game with two
players, say Alice and Bob, who are unable to communicate while the game is in
progress. Alice is given a constraint at
random and must return a satisfying assignment to the variables in the
constraint. Bob is given a variable from the constraint and must return an
assignment to the variable. The two players win if they assign the same value to
the chosen variable. With shared randomness, Alice and Bob can play the game
perfectly if and only if the instance is satisfiable~\cite{Cleve14:icalp}.
The Mermin-Peres
construction~\cite{Mermin1990simple,Peres1990incompatible}
was the first example of an instance where the players can play perfectly by
sharing an entangled quantum state although the instance is not satisfiable. We
note that~\cite{Mermin1990simple,Peres1990incompatible} were looking at quantum contextuality scenarios rather than
non-local games and it was Aravind who reformulated the construction in the above-described
game setting~\cite{aravind2002bell}, cf. also~\cite{Cleve04:ccc}. The game was studied systematically for Boolean
CSP by Cleve and Mittal~\cite{Cleve14:icalp}
and further studied by Cleve, Liu, and Slofstra~\cite{Cleve17:jmp-perfect}.

Every Boolean relation can be identified with its characteristic function, which
has a unique representation as a multilinear polynomial via the Fourier
transform. The multilinear polynomial representation of Boolean relations (and
thus also Boolean CSPs) makes it possible to consider \emph{relaxations} of
satisfiability in which the variables take values in a suitable space,
rather than in $\{-1,+1\}$. Such relaxations of satisfiability have been
considered in the  foundations of physics long ago, playing an important role in
our understanding of the differences between classical and quantum theories.
In detail, given a Boolean CSP instance, a classical assignment assigns to every
variable a value from $\{-1,+1\}$. An operator assignment assigns to every
variable a linear operator $A$ on a Hilbert space so that $A^2=I$
and each $A$ is self-adjoint, i.e., $A=A^*$ and thus in particular $A$ is
normal, meaning that $AA^*=A^*A$.\footnote{For finite-dimensional Hilbert
spaces, a linear operator is just, up to a choice of basis, a matrix with
complex entries, and the adjoint $A^*$ of $A$ is the conjugate transpose of $A$,
also denoted by $A^*$. For
infinite-dimensional Hilbert spaces, the notions are more involved, cf.~\Cref{sec:preliminaries} for all details.}
Furthermore, it is required that operators
assigned to variables from the scope of some constraint should pairwise commute. 

Ji showed that for Boolean CSPs corresponding to \textsc{2-SAT} there is no difference between (classical)
satisfiability and satisfiability via operators~\cite{Ji13:arxiv-binary}.
Later, Atserias, Kolaitis, and Severini established a complete classification for all Boolean CSPs parameterized by the set of
allowed constraint relations. In particular, using the
substitution method~\cite{Cleve14:icalp} they showed that
that only CSPs whose
relations are 0-valid, 1-valid, or come from \textsc{2-SAT}, \textsc{Horn-SAT}, or \textsc{Dual
Horn-SAT} have ``no satisfiability gap'' in the sense that (classic)
satisfiability is equivalent to operator satisfiability (over both finite- and
infinite-dimensional Hilbert spaces)~\cite{AKS19:jcss}. 
For all other Boolean CSPs, Atserias et al.~\cite{AKS19:jcss} showed that there are
satisfiability gaps in the
following sense: There are instances that are not (classically) satisfiable but are
satisfiable via operators on finite-dimensional Hilbert spaces, relying on the Mermin-Peres magic
square --- gaps of the \emph{first kind}.  This immediately implies that
there are instances that are not (classically) satisfiable but are satisfiable
via operators on infinite-dimensional Hilbert spaces --- gaps
of the \emph{second kind}. Finally, there are instances that are not satisfiable via
operators on finite-dimensional Hilbert spaces but are satisfiable via
operators on infinite-dimensional Hilbert spaces, relying on the 
breakthrough result of Slofstra who proved this for linear
equations~\cite{Slofstra20:jams} --- gaps of the \emph{third kind}.
The result is
established in~\cite{AKS19:jcss} by showing that reductions between CSPs based on primitive positive
formulas, which preserve complexity and were used to establish Schaefer's
classification of Boolean CSPs~\cite{Schaefer78:stoc}, preserve satisfiability gaps.

\paragraph{Contributions}

As our main contribution, we generalize the result of Atserias et
al.~\cite{AKS19:jcss} from Boolean CSPs to CSPs on \emph{arbitrary finite}
domains.
As has been done in, e.g,~\cite{Culf24:focs,Goemans04:jcss,Qassim20:jpa}, we represent a
finite domain of size $d$ by the $d$-th roots of unity, and require that each
operator $A$ in an operator assignment should be normal (i.e. $AA^*=A^*A$) and
should satisfy $A^d=I$. The representation of non-Boolean CSPs relies on
multi-dimensional Fourier transform. Our main result is a partial classification
of satisfiability gaps for CSPs on arbitrary finite domains, with the remaining
cases being related to a well-known open problem in the area, which is the
existence of gap instances of the first kind for linear equations over
$\mathbb{Z}_p$ with $p>2$. 

In detail, as our first result we prove that 
\emph{NP-hard} CSPs do have gaps of all three kinds. Second, we show that CSPs of \emph{bounded width} (on arbitrary finite domains) do not have a
satisfiability gap (of any kind), meaning that classical satisfiability is
equivalent to satisfiability via operators on finite- and infinite-dimensional
Hilbert spaces (Theorem~\ref{the:no-gap}). Third, all other CSPs (i.e.,
tractable CSPs of unbounded width) can simulate, in a
precise technical sense, linear equations over an Abelian group of prime order
$p$. We prove that these simulations preserve
satisfiability gaps (Theorem~\ref{the:hsp-gap}).
Thus, we obtain a satisfiability
gap of the second kind by the result by Slofstra and
Zhang~\cite{SZ24:arxiv}.
Moreover, for $p=2$, we obtain a satisfiability
gap of the first kind (and thus also of the second kind) from the Mermin-Peres magic
square~\cite{Mermin1990simple,Mermin1993hidden,Peres1990incompatible}, and a 
satisfiability gap of the third kind from Slofstra's
result~\cite{Slofstra20:jams} (cf. also~\cite{kim2018synchronous} for a
different proof of Slofstra's result). Finally, the reductions will also establish the
first claim, gaps of all three kinds for NP-hard CSPs
via~\cite{Mermin1990simple,Mermin1993hidden,Peres1990incompatible,Slofstra20:jams}.

\begin{theorem}[Main result, informal statement]\label{thm:main-informal}
  Let $\Gm$ be an arbitrary finite set of relations on a finite domain. 
  \begin{enumerate}
  \item If $\CSP(\Gm)$ is NP-hard then $\CSP(\Gm)$ has gaps of all three kinds.
  \item Otherwise, $\CSP(\Gm)$ is solvable in polynomial time and
  \begin{enumerate}
  \item either $\Gm$ has bounded width and $\CSP(\Gm)$ has no gaps of any kind,
  \item or $\CSP(\Gm)$ does not have bounded width and $\CSP(\Gm)$ can simulate linear equations over an Abelian
    group of prime order $p$ in a satisfiability-gap-preserving fashion. In this
    case, $\CSP(\Gm)$ has a gap of the second kind. Furthermore, if $p=2$ then
    $\CSP(\Gm)$ has gaps of all three kinds. 
  \end{enumerate} 
  \end{enumerate}
\end{theorem}

We note that Theorem~\ref{thm:main-informal} implies a classification of Boolean
CSPs, thus recovering the result by Atserias et al.~\cite{AKS19:jcss}. Indeed,
a Boolean language $\Gm$ that is 0-valid, 1-valid, or comes from
\textsc{2-SAT}, \textsc{Horn-SAT}, or \textsc{Dual Horn-SAT} has bounded
width; also, for a Boolean $\Gm$ one gets $p=2$ in
Theorem~\ref{thm:main-informal}\,(2b).

Theorem~\ref{thm:main-informal} also implies
a classification for CSPs on graphs (such CSPs are also known as
\textsc{$H$-Coloring} problems): If $\Gm=\{\rel\}$, where $\rel$ is a binary
symmetric irreflexive relation, then either $R$ is bipartite or not. In the former case,
$\CSP(\Gm)$ has bounded width~\cite{Bulatov18:survey} and thus has no gaps of any kind. In the latter
case, $\CSP(\Gm)$ is NP-hard~\cite{HellN90,Bulatov05:tcs} and thus has gaps of
all three kinds.

\begin{corollary}\label{cor:graphs}
Let $G=(V,E)$ be a graph. If $G$ is bipartite, $\CSP(\{E\})$ has no satisfiability gap of any kind. Otherwise it has satisfiability gaps of all three kinds.
\end{corollary}

We note that the remaining cases not covered by Theorem~\ref{thm:main-informal}
(i.e., $\CSP(\Gm)$ that are tractable but have unbounded width and can simulate
linear equations over an Abelian group of prime order $p>2$ but not of order
$2$) relate to the known difficulties of establishing satisfiability gaps of the
first kind for linear equations over $\mathbb{Z}_p$ with $p>2$, e.g., there are
results showing that such an instance cannot be obtained from generalized Pauli
matrices~\cite{Qassim20:jpa} and beyond~\cite{Frembs22:arxiv,chung2024simplicial}.\footnote{We
remark that our results show that one can obtain gap instances from homomorphic
preimages of already established gap instance (cf. Step~4 in~\Cref{sec:gap}).}

\medskip
The proof of Theorem~\ref{thm:main-informal} relies on several ingredients. Firstly, we observe that
results establishing that primitive positive definability preserve
satisfiability gaps~\cite{AKS19:jcss} can be lifted from Boolean to arbitrary finite
domains.\footnote{\cite[Theorem~6.4]{Mastel24:stoc} considers subdivisions, a
special case of conjunctive definitions, which in turn are a special case of
primitive positive definitions (without existential quantifiers). Our result
shows that~\cite[Theorem~6.4]{Mastel24:stoc} extends to conjunctive
definitions.}
Secondly, for CSPs of bounded width we show that there is no difference
between classical and operator satisfiability by simulating the inference by the
so-called Singleton Linear Arc Consistency (SLAC) algorithm in polynomial
equations. We note that while there are several (seemingly stronger) algorithms
for CSPs of bounded width, our proof relies crucially on the special structure of SLAC and
the breakthrough result of Kozik that SLAC solves all CSPs of bounded
width~\cite{Kozik21:sicomp}. Thirdly, to prove that NP-hard CSPs and certain CSPs of unbounded width have
satisfiability gaps we use the algebraic approach to CSPs, namely, we show that not
only primitive positive definitions but also other reductions,
namely going to the core, adding constants, and restrictions to subalgebras and
factors, preserve satisfiability gaps.

We note that there is a significant hurdle to go from Boolean CSPs to CSPs
over arbitrary finite domains. While any non-Boolean CSP can be Booleanized via
indicator variables and extra constraints, such constructions do not
immediately imply classifications of non-Boolean CSPs as the ``encoding
constraints'' are intended to be used in only a particular way. Indeed, while
the complexity of Boolean CSPs was established by Schaefer in
1978~\cite{Schaefer78:stoc} and the complexity of CSPs on three-element domains
was established by Bulatov in 2002~\cite{Bulatov02:focs,Bulatov06:jacm}, the
dichotomy for all finite domains was only established in
2017~\cite{Bulatov17:focs,Zhuk17:focs,Zhuk20:jacm}.
Similarly for other variants of CSPs and different notions of tractability,
results on Boolean domains, including the work of Atserias et
al.~\cite{AKS19:jcss}, rely crucially on the explicit knowledge of the structure
of relations on Boolean domains (established by Post~\cite{Post41}), which is not known for non-Boolean CSPs.
Indeed, on the tractability side, the structure of relations in tractable Boolean CSPs is simple and very well understood; on the intractability side, reductions based only on primitive positive definitions suffice for a complete classification of Boolean CSPs. Neither of these two facts is true for non-Boolean CSPs.

We find it fascinating that bounded width might be the borderline for
satisfiability gaps for CSPs, thus linking a notion coming from a
natural combinatorial algorithm for CSPs with a foundational topic in quantum
computation. 
Indeed, if satisfiability gaps of the first kind exist for linear equations over
$\mathbb{Z}_p$ with $p>2$ then our main result implies that all CSPs of
unbounded width admit satisfiability gaps of all three kinds.
This is yet another result indicating the fundamental nature of bounded
width, which captures not only the power of the local consistency
algorithm~\cite{LaroseZadori07:au,Maroti2008existence,Bulatov09:width,Barto14:local}
as conjectured in~\cite{Feder98:monotone} with links to Datalog, pebble games,
and logic~\cite{Feder98:monotone,Kolaitis00:jcss}, but also robust solvability
of CSPs~\cite{Barto16:sicomp}, exact solvability of valued CSPs by
LP~\cite{tz17:sicomp} and SDP~\cite{tz18} relaxations, and now also
possibly satisfiability via operators.

\paragraph{Related work}
Paddock and Slofstra~\cite{Paddock25:ijm} 
streamlined the results from~\cite{AKS19:jcss}, and also gave 
an overview of other notions of
satisfiability and their relationship, including the celebrated MIP$^*$=RE
result of Ji et al.~\cite{Ji20:arxiv,Ji22:cacm}.

There is some very recent work on approximability of operator CSPs, initiated by
 Culf, Mousavi, and Spirig~\cite{Culf24:focs}. Mousavi and Spirig studied a quantum analog of the
unique games conjecture~\cite{Mousavi25:itcs}, and Man\v{c}inska, Spaas, Spirig and Vernooij studied gap-preserving
approximation reductions~\cite{Mancinska25:focs}

While we used the magic square to construct a gap of the first kind, other
systems based on linear equations could be used instead, e.g. the pentagram and
the BCS in~\cite{Slofstra20:jams}, cf. also~\cite{Zhang24:nature}.
Finally, we note that the notion of quantum homomorphism from the work of Man\v{c}inska and
Roberson~\cite{Mancinska16:jctb}, as well as the recent work of
Ciardo~\cite{Ciardo24:lics}, is different from ours except when $d=2$. In the
case of $d=2$, as considered in~\cite{AKS19:jcss}, operator assignments have to
be self-adjoint and can be transformed to quantum homomorphisms. When $d\ge3$,
the straightforward translation through Fourier transform in cyclic groups
result in projectors that are not self-adjoint, and therefore do not give rise
to a quantum homomorphism. 
However, the overall results are somewhat similar, e.g., Ciardo shows~\cite{Ciardo24:lics} that CSPs of
bounded width have no quantum advantage (over finite-dimensional spaces) by establishing a so-called minion
homomorphism --- relying crucially on idempotency --- to the minion capturing the basic SDP relaxation, which is known
to solve CSPs of bounded width by the work of Barto and Kozik~\cite{Barto16:sicomp}. The methods used
in~\cite{Ciardo24:lics} rely, unlike methods and results in this paper, on being in the finite-dimensional case.

\paragraph{Paper structure}

After introducing all necessary notation and definitions
in~\Cref{sec:preliminaries}, we will in~\Cref{sec:operator-CSP} define operator
CSPs formally. Then, in~\Cref{sec:overview} we will present our results, explain
the intuition behind the proofs by sketching some of the main ideas. The rest of
the paper will then give full details of all proofs. The reader is encouraged to
read~\Cref{sec:overview} before diving into the details in~\Cref{sec:no-gap},
dealing with CSPs of bounded width, \Cref{sec:operator-pp}, proving properties
of primitive positive definitions, and finally \Cref{sec:gap}, dealing with CSPs
of unbounded width.

\section{Preliminaries}\label{sec:preliminaries}


\paragraph{CSPs}
An instance of the \emph{constraint satisfaction problem} (CSP) is a triple
$\cP=(V,D,\cC)$, where $V$ is a set of variables, $D$ is a set of domain values,
and $\cC$ is a set of constraints. Every constraint in $\cC$ is a pair
$\ang{\bs,\rel}$, where $\bs\in V^r$ is the constraint scope and $\rel\subseteq
D^r$ is the constraint relation of arity $r=ar(R)$. Given a CSP instance $\cP$, the
task is to determine whether there is an assignment $s:V\to D$ that assigns to
every variable from $V$ a value from $D$ in such a way that
all the constraints are satisfied; i.e.,
$(s(v_1),\ldots,s(v_r))\in\rel$ for every constraint
$\ang{(v_1,\ldots,v_r),\rel}\in\cC$. An assignment satisfying all the
constraints is also called a \emph{solution}.

Let $D$ be a fixed finite set. A finite set $\Gm$ of relations over $D$ is
called a \emph{constraint language} over $D$. We denote by $\CSP(\Gm)$ the
class of CSP instances in which all constraint relations belong to $\Gm$. A mapping $\vr:D\to D$ is an \emph{endomorphism} or \emph{unary polymorphism} of $\Gm$ if, for any $\rel\in\Gm$ (say, $r$-ary) and any $(\vc ar)\in\rel$, the tuple $(\vr(a_1)\zd\vr(a_r))$ belongs to $\rel$.

\paragraph{Bounded width}

Intuitively, CSPs of bounded width are those CSPs  for which unsatisfiable
instances can be refuted via local propagation. An obvious obstruction to bounded width, in
addition to NP-hard CSPs, is CSPs encoding systems of linear
equations~\cite{Feder98:monotone}. A celebrated result of Barto and Kozik
established that CSPs of bounded width are precisely those CSPs that cannot
simulate, in a precise sense, linear equations~\cite{Barto14:local}. 
While bounded width has several
characterizations~\cite{LaroseZadori07:au,Maroti2008existence,Bulatov09:width,Barto14:local,Kozik15:au},\footnote{One characterization implies that checking whether a constraint language has bounded
width is decidable~\cite{Kozik15:au}.} we will rely on the result of
Kozik~\cite{Kozik21:sicomp} that established that every CSP of bounded width can
be solved through constraint propagation of a very restricted type, so-called
\emph{Singleton Linear Arc-Consistency} (SLAC). 

In order to explain SLAC, we need to start with
\emph{Arc-Consistency} (AC). AC is one of the basic levels of local consistency
notions. It is a property of a CSP and also an algorithm
turning an instance $\cP\in\CSP(\Gm)$ into an equivalent subinstance
$\cP'\in\CSP(\Gm)$ that satisfies the AC property. Intuitively, given an instance $\cP=(V,D,\cC)\in\CSP(\Gm)$,
the AC algorithm starts with setting the domain $D_v=D$ for every variable $v\in
V$. Then, it prunes the sets $\{D_v\}_{v\in V}$ in an iterative fashion,
terminating (in polynomial time in the size of $\cP$) with a maximal subinstance
of $\cP$ that satisfies the AC condition; namely, for every variable $v\in V$, every value $a\in
D_v$, and every constraint $\ang{\bs,\rel}\in\cC$ such that $\bs[i]=v$ for some $i$, there is a tuple $\ba\in\rel$ with
$\ba[i]=a$. 
The resulting subinstance $\cP'$ is \emph{equivalent} to $\cP$
in the sense that $\cP$ has a solution if and only if $\cP'$ has a solution. We
say that AC \emph{solves} an instance $\cP$ if $\cP$ has a solution whenever
$\cP'$ is consistent; i.e., none of the sets $D_v$ in $\cP'$ is empty.
AC is not strong enough to solve all CSPs of bounded width (e.g., \textsc{2-SAT}) but its
full power is understood~\cite{Feder98:monotone,Dalmau99}.

Equivalently, Arc-Consistency can be described in terms of a Datalog
program~\cite{Kolaitis95:jcss}. In general, a Datalog program derives facts about an
instance $\cP\in\CSP(\Gm)$
using a fixed set of rules that depend on the constraint language $\Gm$. The
rules are defined using relations from $\Gm$ called
\emph{extensional databases} (EDBs) as well as a number of auxiliary relations
called \emph{intensional databases} (IDBs). Each rule consists of a \emph{head},
which is a single IDB, and the \emph{body}, which is a sequence of IDBs and
EDBs. The execution of the program updates the head IDB whenever the body of the
rule is satisfied, that is, every EDB and IDB in the body is satisfied. The computation of the program ends when no relation can be
updated, or when the goal predicate is reached. If we require that a Datalog
program should only include unary IDBs and that every rule should have at most one EDB
then the power of the program for CSPs amounts to AC. 
In detail, the AC Datalog program has a unary relation (IDB) $T_S(v)$ for each 
subset $S\sse D$. Then for every $\ang{(\vc vr),\rel}\in\cC$ and for any
IDBs $T_{S_1}\zd T_{S_m}, T_S$ the program contains the rule
\[
T_S(v_i):-\rel(\vc vr),T_{S_1}(v_{i_1})\zd T_{S_m}(v_{i_m})
\]
if for any $\ba\in\rel$ such that $\ba[i_j]\in S_j$ we have $\ba[i]\in S$.
The Arc-Consistency algorithm is \emph{Linear} if $m=1$ for every rule in the
corresponding Datalog program. 

The Singleton Arc-Consistency (SAC) algorithm is a modification of
the AC algorithm~\cite{DB97}. SAC updates the sets $\{D_v\}_{v\in V}$
as follows: it removes $a$ from $D_v$ if the current
instance augmented with the constraint fixing the value $a$ to the variable $v$ is
found inconsistent by the AC algorithm. 
Finally, the Singleton Linear Arc-Consistency algorithm is a modification of SAC
(due to Kozik~\cite{Kozik21:sicomp} and Zhuk~\cite{Zhuk20:jacm}) that uses the
Linear AC algorithm rather than the AC algorithm. Kozik has shown that SLAC
solves all CSPs of bounded width~\cite{Kozik21:sicomp}.
As with AC, SLAC is not only an algorithm but also a condition (of the instance
$\cP'$ produced by the algorithm). We say that an instance
$\cP$ is SLAC-consistent if the SLAC algorithm, given in Figure~\ref{fig:slac},
does not change the instance.

\begin{figure}
\begin{itemize}\setlength\itemsep{-3pt}
\item[Input:]
A CSP instance $\cP=(V,D,\cC)$.
\item[Output:]
A SLAC-consistent instance $\cP'$ equivalent to $\cP$
\item[1.]
{\bf for each} $v\in V$ {\bf set} $D_v=D$
\item[2.]
$\cP'=\cP+\sum_{v\in V}\ang{v,D_v}$
\item[3.]
{\bf until} the process stabilizes
\begin{itemize}
\item[3.1]
{\bf pick} a variable $v\in V$
\item[3.2]
{\bf for each} $a\in D_v$ {\bf do}
\begin{itemize}
\item[3.2.1]
{\bf run} Linear Arc-Consistency on $\cP'+\ang{v,\{a\}}$
\item[3.2.2]
{\bf if} the problem is inconsistent, {\bf set} $D_v=D_v-\{a\}$
\end{itemize}
\item[]
{\bf endfor}
\end{itemize}
\item[]
{\bf enduntil}
\item[4.]
{\bf return} $\cP'$
\end{itemize}
\caption{SLAC.}\label{fig:slac}
\end{figure}

\paragraph{Multi-dimensional Fourier transform}\label{sec:multi-FT}

Let $U_d$ be the set of $d$-th roots of unity, that is,
$U_d=\{\ld_k=e^{\frac{2\pi i}d k}\mid 0\le k<d\}$.
The \emph{Fourier transform} (FT) of a function $f:U_d^n\to U_d$ is defined, for 
$\ba\in U_d^n$, as
$
FT(f,\ba)=\sum_{\bb\in U_d^n}f(\bb)\ld_1^{\ba\cdot\bb}$.
Then it is not hard to see that 
$
f(\ba)=\sum_{\bb\in U_d^n}FT(f,\bb)\ld_1^{\ba\cdot\bb}$,
which gives rise to a representation of $f$ by a polynomial
\[
f(\ov x)=\sum_{\bb\in U_d^n}FT(f,\bb)\ov x^{\bb'},
\]
where $\bb'=(\vc kn)$ and $\bb[j]=\ld_{k_j}$. This representation is
unique~\cite{ODonnell14:book}.

\paragraph{Linear operators and Hilbert spaces}
Let $V$ be a complex vector space. A \emph{linear operator} on $V$ is a linear
map from $V$ to $V$. The identity linear operator on $V$ is denoted by $I$. The
linear operator that is identically $0$ is denoted by $0$. 
Let $A$ and $B$ be two linear operators. Their pointwise addition is denoted by $A+B$, their
composition is denoted by $AB$, and the pointwise scaling of $A$ by a scalar
$c\in\zC$ is denoted by $cA$. All of these are linear operators and thus we can
plug linear operators in a polynomial $P$. 
We say that operators $A$ and $B$ \emph{commute} if
$AB=BA$. A collection of linear operators $A_1,\ldots,A_n$ \emph{pairwise commute} if $A_iA_j=A_jA_i$ for every $i,j\in [n]$. 
If $A_1,\ldots,A_n$ pairwise commute then $P(A_1,\ldots,A_n)$ is well defined.
$\zC[x_1,\ldots,x_n]$ denotes the ring of polynomials with complex coefficients and commuting variables $x_1,\ldots,x_n$. 
A linear operator is \emph{bounded} if it maps bounded subsets to bounded
subsets. Let $A$ be a densely defined linear operator, i.e., $A$ is defined
almost everywhere. We denote by $A^*$ its \emph{adjoint}~\cite{Folland94}
and call $A$ \emph{normal} if $A$ commutes with its adjoint, i.e., $AA^*=A^*A$.
A linear operator $U$ is called \emph{unitary} if $UU^*=U^*U=I$, the identity
operator.

A \emph{Hilbert space} is a complex vector space with an inner product whose
norm induces a complete metric.
All Hilbert spaces of finite dimension $d$ are
isomorphic to $\zC^d$ with the standard complex inner product. Thus, 
after the choice of basis, linear operators on a $d$-dimensional Hilbert space
can be identified with matrices in $\zC^{d\times d}$, and operator composition
becomes matrix multiplication. 
Thus, for Hilbert spaces of finite dimension we will
freely switch between the operator and matrix terminology.
A \emph{diagonal} matrix has all off-diagonal
entries equal to $0$. 
For a matrix $A$, the adjoint operator $A^*$ is the conjugate transpose.
Recall that $(AB)^*=B^*A^*$.
We will use the following form of the so-called \emph{Strong Spectral
Theorem}.
\begin{theorem}[\cite{Halmos2017introduction}]
\label{the:SST}
Let $\vc Ar$ be normal matrices. If $\vc Ar$ pairwise
commute then there exists a unitary matrix $U$ and diagonal matrices 
$\vc Er$ such that $A_i=U^{-1}E_iU$ for every $i\in[r]$.
\end{theorem}

In the infinite-dimensional case, the Strong Spectral Theorem is more
complicated. Firstly, it involves general $L^2$- and $L^\infty$-spaces. Let
$(\Omega,\cM,\mu)$ be a measure space, i.e., $\Omega$ is a set, $\cM$ is a
$\sigma$-algebra on the set $\Omega$ (i.e. $\cM$ is a nonempty collection of
subsets of $\Omega$ containing $\Omega$ and closed under complements, countable intersections, and
countable unions), and $\mu$ is a measure on $(\Omega,\cM)$ (i.e. $\mu$ is
nonnegative and countably additive). Then $L^2(\Omega,\mu)$ denotes
the collection of square integrable measurable functions, up to almost
everywhere equality, and $L^\infty(\Omega,\mu)$ denotes the
collection of essentially bounded measurable functions, up to almost everywhere equality; all measure-theoretic terms in
these definitions refer to $\mu$, cf.~\cite{Folland94} for details. Instead of diagonal matrices, the General Strong Spectral Theorem uses \emph{multiplication operators} on a $L^2(\Omega,\mu)$-space. Let $V$ be a complex vector space of functions mapping an index set $X$ to $\zC$. A multiplication operator of $V$ is a linear operator whose value at a function $f:X\to\zC$ in $V$ is given by pointwise multiplication by a fixed function $a:X\to\zC$. In other words, the multiplication operator given by $a$ is
\[
(T_a(f))(x)=a(x)f(x)\qquad\text{for each $x\in X$}.
\]

\begin{theorem}[General Strong Spectral Theorem~\cite{Folland94}]\label{the:general-sst}
Let $\vc Ar$ be normal bounded linear operators on a Hilbert space $\cH$. If $\vc Ar$ pairwise commute then there exist measure space $(\Omega,\cM,\mu)$, a unitary map $U:\cH\to L^2(\Omega,\mu)$, and functions $\vc ar\in L^\infty(\Omega,\mu)$ such that $A_i=U^{-1}T_{a_i}U$ for every $i\in[r]$.
\end{theorem}

\section{Operator CSP}\label{sec:operator-CSP}

In order to relax the notion of satisfiability, we first consider CSPs on $U_d$
for some $d$ and represent CSPs via polynomials. Let $\Gm$ be a constraint 
language over $U_d$. Every constraint $\ang{\bs,\rel}$ of an instance 
$\cP=(V,U_d,\cC)$ of $\CSP(\Gm)$ is represented by a polynomial 
$P_\rel(\bs)$ that represents the characteristic function $f_\rel$ of
$\rel$:
\[
f_\rel(\bs)=\left\{\begin{array}{ll}
  \ld_0, & \text{if $\rel(\bs)$ is true,}\\
  \ld_1, & \text{otherwise}.
\end{array}\right.
\]
We note that our choice of the polynomial representation is somewhat arbitrary
but other choices lead to the same results. (For instance, \cite{AKS19:jcss}
studied the $d=2$ case and used $\ld_1$ to represent true.) 

An operator $A$ on a Hilbert space $\cH$ is a \emph{normal operator of order
$d$} if $A$ is normal and $A^d=I$. Operator assignment to the instance $\cP$ on a Hilbert space $\cH$ is a mapping 
that assigns to every variable from $V$ an operator $A_v$ on $\cH$ such that
\begin{itemize}
\item[(a)]
$A_v$ is a normal operator of order $d$ for every $v\in V$;
\item[(b)]
the operators $A_{v_1}\zd A_{v_r}$ pairwise commute
for every constraint  $\ang{(\vc vr),\rel}\in\cC$.
\end{itemize}

\noindent
We call an operator assignment $\{A_v\}$ \emph{satisfying for $\cP$} if 
$P_\rel(A_{v_1}\zd A_{v_r})=I$ for every constraint $\ang{(\vc vr),\rel}\in\cC$.
Let $\cP$ be a CSP instance. 
Following the terminology of Atserias et al.~\cite{AKS19:jcss}, we say that
$\cP$ has a \emph{satisfiability gap of the first kind} if $\cP$ is not
satisfiable over $U_d$ but is satisfiable by an operator assignment on 
a finite-dimensional Hilbert space. Similarly, we say that 
$\cP$ has a \emph{satisfiability gap of the second kind} if $\cP$ is not
satisfiable over $U_d$ but is satisfiable by an operator assignemnt on an
infinite-dimensional Hilbert space. Finally, we say that 
$\cP$ has a \emph{satisfiability gap of the third kind} if $\cP$
is not satisfiable on finite-dimensional Hilbert spaces but is satisfiable by an
operator assignment on 
an infinite-dimensional Hilbert space.
We say that $\CSP(\Gm)$ has a satisfiablity gap of the $i$-th
kind, $i=1,2,3$, if there is at least one instance $\cP\in\CSP(\Gm)$ witnessing
such a gap.
By definition, a gap of the first kind implies a gap of the second kind. Also, a
gap of the third kind implies a gap of the second kind. Put differently, if
$\Gm$ has no gap of the second kind then it has no gap of the first or third
kind either.

We shall repeatedly use the following simple lemma.

\begin{lemma}\label{lem:lemma-3}
Let $\vc xr$ be variables, let $\vc Qm,Q$ be polynomials in $\zC[\vc xr]$,
and let $\cH$ be a Hilbert space. If every assignment over $U_d$ that 
satisfies the equations $Q_1=\dots=Q_m=0$ also satisfies the equation
$Q=0$, then every fully commuting operator assignment on $\cH$ that 
satisfies the equations $Q_1=\dots=Q_m=0$ also satisfies the equation
$Q=0$.
\end{lemma}

The proof of this lemma is very similar to that of the analogous claim
in~\cite[Lemma~3]{AKS19:jcss}, where it was established for $d=2$.
The main difference is to 
use $A^d=I$ rather than $A^2=I$. For the sake of completeness we give the proof here.

\noindent\begin{proof}
{\it Finite-dimensional case.}
Suppose that the conditions of the lemma hold and $\vc Ar$ are pairwise
  commuting operators such that the equations $Q_1=\dots=Q_m=0$ are
  true when these matrices are assigned to $\vc xr$. Then, since $\vc Ar$ are
  normal and commute, by Theorem~\ref{the:SST} there is a unitary matrix $U$
  such that $E_i=UA_iU^{-1}$ is a diagonal matrix. Then, $E_i^d=I$, because $A_i^d=I$. Therefore, every diagonal entry $E_i(jj)$ belongs to $U_d$. For every equation $Q_\ell$ we have $Q_\ell(\vc Ar)=0$ implying
  $Q_\ell(\vc Er)=UQ_\ell(\vc Ar)U^{-1}=0$. Since every $E_i$ is diagonal, for every $j$ it also holds $Q_\ell(E_1(jj)\zd E_r(jj))=0$. By the conditions of the lemma we also have $Q((E_1(jj)\zd E_r(jj))=0$, and $Q(\vc Ar)=U^{-1}Q(\vc Er)U=0$.

\smallskip

{\it Infinite-dimensional case.}
Suppose that the conditions of the lemma hold and $\vc Ar$ are pairwise commuting  operators such that the equations $Q_1=\dots=Q_m=0$ are true when these operators are assigned to $\vc Xr$. Then, since $\vc Ar$ are normal and commute, by  Theorem~\ref{the:general-sst} there exist a measure space $(\Omega,\cM,\mu)$, a unitary map $U:\cH\to L^2(\Omega,\mu)$, and functions $\vc ar\in L^\infty(\Omega,\mu)$ such that, for the multiplication operators $E_i=T_{a_i}$ of $L^2(\Omega,\mu)$, the equalities $A_i=U^{-1}E_iU$ hold for $i\in[r]$. This implies $UA_iU^{-1}=E_i$. As $A_i^d=I$ we also have $E_i^d=I$. Therefore $a_i(\omega)^d=1$ for almost all $\omega\in\Omega$, or more formally $\mu(\{\omega\in\Omega\mid (a_i(\om)^d\ne1\})=0$. Hence $a_i(\omega)\in U_d$ for almost all $\omega\in\Omega$. For every $\ell\in[m]$ we have $Q_\ell(\vc Er)=UQ_\ell(\vc Ar)U^{-1}$, implying by the conditions of the lemma that $Q_\ell(\vc Er)=0$. Since $E_i$ is the multiplication operator given by the function $a_i$, it holds that $Q_\ell(a_1(\om)\zd a_r(\om))=0$ for almost all $\om\in\Omega$. As for almost all $\om\in\Omega$, it holds that the value $a_i(\om)\in U_d$ for each $i\in[r]$, we have $Q(a_1(\om)\zd a_r(\om))=0$ for almost all $\om\in\Omega$. Therefore $Q(\vc Er)=0$ and this implies $Q(\vc Ar)=0$ as before.
\end{proof}

\section{Overview of results and techniques}\label{sec:overview}

In this section, we give an overview of how our main result is proved. All
definitions and details are provided in the main part of the paper comprizing of~\Cref{sec:no-gap}--\Cref{sec:gap}.

\paragraph{Bounded width}

One part of our main result is the following.

\begin{theorem}\label{the:no-gap-o}
Let $\Gm$ be a constraint language over $U_d$. If $\CSP(\Gm)$ has bounded
width then it has no satisfiability gap of any kind.
\end{theorem}

The main idea behind the proof of Theorem~\ref{the:no-gap-o} 
is to simulate the inference provided by SLAC (see Fig.~\ref{fig:slac}) 
by inference in polynomial equations. Let $\cS$ be a SLAC-program solving $\CSP(\Gm)$.
In order to prove Theorem~\ref{the:no-gap-o} we take an instance $\cP=(V,U_d,\cC)$ of $\CSP(\Gm)$ that has no solution, and therefore  is not SLAC-consistent, as $\CSP(\Gm)$ has bounded width, and prove that it also has no operator solution. We will prove it by contradiction, assuming $\cP$ has an operator solution $\{A_v\}$ and then using the rules of a SLAC-program solving $\CSP(\Gm)$ to infer stronger and stronger conditions on $\{A_v\}$ that eventually lead to a contradiction. 

Recall that every rule of a SLAC-program has the form $(x\in S)\meet\rel(x,y,\vc zr))\to (y\in S')$ for some variables $x,y\in V$, a constraint $\ang{(x,y,\vc zr),\rel}$, and sets $S,S'\sse U_d$. Therefore, we need to show how to encode unary relations and rules of a SLAC-program through polynomials. For any $S\sse U_d$, the unary constraint restricting the domain of a variable
$x$ to the set $S$ is represented by the polynomial
\[
Dom_S(x)=\prod_{k\in S}(\ld_k-x)+1.\footnote{This is not the representation of $S$ as in the beginning of~\Cref{sec:operator-CSP}, as $Dom_S(a)$ is not necessarily equal to $\ld_1$ for $A\not\in S$. However, it suffices for our purposes, because we only need the property that $Dom_S(a)=1$ if and only if $a\in S$.}
\]
Similarly, the rule $(x\in S)\meet\rel(x,y,\vc zr))\to (y\in S')$ of the SLAC 
program is represented by
\[
Rule_{S,\rel,S'}(x,y,\vc zr)=(Dom_{\ov S}(x)-1)(P_\rel(x,y,\vc zr)-\ld_1)
(Dom_{S'}(y)-1).
\]
To give an idea of how Theorem~\ref{the:no-gap-o} is proved, we sketch the proof of the following.
\begin{lemma}\label{lem:transitive-poly-o}
Let $(v_1\in S_1)\to\dots\to(v_\ell\in S_\ell)$ be a derivation in the SLAC-program $\cS$ and $\{A_v\}$ an operator assignment for $\cP$. 
Then for each $i=2\zd\ell$
\[
(Dom_{\ov S_1}(A_{v_1})-I)(Dom_{S_i}(A_{v_i})-I)=0.
\]
\end{lemma}
\begin{proof}[Proof Sketch]
First, one shows via Lemma~\ref{lem:lemma-3} that any operator assignment is a zero of $Rule_{S,\rel,S'}$
(cf. Lemma~\ref{lem:rule-poly}). This can be used to establish the claim for
  $i=2$ (cf. Lemma~\ref{lem:derivation-poly}).
The rest of the proof is done by induction on $i$. 
In the inductive case we have equations
\begin{equation}\label{equ:equation1-o}
(Dom_{\ov S_1}(A_{v_1})-I)(Dom_{S_i}(A_{v_i})-I)=0,
\end{equation}
and 
\begin{equation}\label{equ:equation2-o}
(Dom_{\ov S_i}(A_{v_i})-I)(Dom_{S_{i+1}}(A_{v_{i+1}})-I)=0.
\end{equation}
The idea is to multiply (\ref{equ:equation1-o}) by 
$(Dom_{S_{i+1}}(A_{v_{i+1}})-I)$ on the right, multiply (\ref{equ:equation2-o})
by $(Dom_{\ov S_1}(A_{v_1})-I)$ on the left and subtract. The problem
is, however, that 
\[
Dom_{S_i}(A_{v_i})-Dom_{\ov S_i}(A_{v_i})
\]
is not a constant polynomial. So, we also need to prove that any polynomial
of the form 
\[
Dom_S(x)-Dom_{\ov S}(x)
\]
is invertible modulo $x^d-1$. The polynomial has the form
\[
p(x)=\prod_{k\in S}(x-\ld_k)-\prod_{k\not\in S}(x-\ld_k).
\]
As is easily seen, if $S\ne U_d$ and $S\ne\eps$, then $\ld_k$ is not a
root of $p(x)$ for any $\ld_k\in U_d$. Therefore the greatest common divisor of $p(x)$ and $x^d-1$ has degree 0, and hence there exists $q(x)$ such that 
\[
p(x)q(x)=c+r(x)(x^d-1).
\]

Thus before subtracting equations (\ref{equ:equation1-o}) and 
(\ref{equ:equation2-o}) we also multiply them by $q(A_{v_i})$. Then we get
\begin{eqnarray*}
(Dom_{\ov S_1}(A_{v_1})-I)(Dom_{S_i}(A_{v_i})q(A_{v_i})
-q(A_{v_i})Dom_{\ov S_i}(A_{v_i}))(Dom_{S_{i+1}}(A_{v_{i+1}})-I) &=& 0\\
(Dom_{\ov S_1}(A_{v_1})-I)q(A_{v_i})(Dom_{S_i}(A_{v_i})
-Dom_{\ov S_i}(A_{v_i}))(Dom_{S_{i+1}}(A_{v_{i+1}})-I) &=& 0\\
c(Dom_{\ov S_1}(A_{v_1})-I)(Dom_{S_{i+1}}(A_{v_{i+1}})-I) &=& 0.
\end{eqnarray*}
The first transformation uses the fact that $A_{v_i}$ commutes with itself, while the 
second one uses the property $A_{v_i}^d=I$.
The result follows.
\end{proof}

To complete the proof of Theorem~\ref{the:no-gap-o} note that the lack of
SLAC-consistency means that for some $v\in V$ the statement 
$(v=\ld_k)\to(v\ne\ld_k)$ can be derived from $\cP$ for every $\ld_k\in U_d$.
By Lemma~\ref{lem:transitive-poly-o}, for any operator assignment $\{A_w\}$
and any $\ld_k\in U_d$ the operator $A_v$ satisfies the equation
$\prod_{j\ne k}(A_v-\ld_jI)=0$.
By reverse induction on the size of $S$, one can show that for any $S\sse U_d$ these equations imply 
$\prod_{j\in S}(A_v-\ld_jI)=0$.
Then for $S=\eps$ we get $I=0$, witnessing that $\cP$ has no satisfying 
operator assignment. 

\paragraph{NP-hard CSPs and unbounded width}

The second part of our main result boils down to the following.

\begin{theorem}\label{the:hsp-gap-o}
Let $\Gm$ be a constraint language over $U_d$. 
If $\CSP(\Gm)$ does not have bounded width then $\CSP(\Gm)$ can simulate linear
  equations over an Abelian group of prime order $p$ in a
  satisfiability-gap-preserving fashion. Moreover, if $\CSP(\Gm)$ is NP-hard
  then one can take $p=2$.
\end{theorem}

Assume that $\CSP(\Gm)$ does not have bounded width. The ``implementation'' of
linear equations in $\CSP(\Gm)$ that preserves gaps will be achieved
in several steps via a 
chain of reductions that lies at the heart of the algebraic approach to CSPs~\cite{Bulatov05:classifying}. 
The reductions are shown in Figure~\ref{fig:reductions-o}. They are
known to preserve satisfiability; we show that they also preserve satisfiability
gaps. 

The most basic reduction (used in the chain in several places) is that of
primitive positive definitions. 
Let $\Gm$ be a constraint language over $U_d$, let $r$ be an integer, and
let $x_1,\ldots,x_r$ be variables ranging over the domain $U_d$. A primitive
positive formula (\emph{pp-formula}) over $\Gm$ is a formula of the form
\begin{equation}\label{eq:pp-formula-o}
  \phi(x_1,\ldots,x_r)=\exists y_1\cdots\exists y_s(\rel_1(\bz_1)\wedge\cdots\wedge
  \rel_m(\bz_m)),
\end{equation}
where $\rel_i$ is either the binary equality relation on $U_d$ or $\rel_i\in\Gm$ is a relation over $U_d$ of arity $r_i$, and each $\bz_i$ is an $r_i$-tuple of variables from
$\{x_1,\ldots,x_r\}\cup\{y_1,\ldots,y_s\}$.
A relation $\rel\subseteq U_d^r$ is primitive positive definable
(\emph{pp-definable}) from $\Gm$ if there
exists a pp-formula $\phi(x_1,\ldots,x_r)$ over $\Gm$ such that $\rel$ is equal to the set of
models of $\phi$, that is, the set of $r$-tuples $(a_1,\ldots,a_r)\in U_d^r$
that make the formula $\phi$ true over $U_d$ if $a_i$ is
substituted for $x_i$ in $\phi$ for every $i\in [r]$.

\begin{figure}[t!]
\[
\CSP(\Gm)\ \leftrightarrow\ \CSP(\core(\Gm)) \leftrightarrow\ \CSP(\core(\Gm)^*)\ \leftarrow\ \CSP(\core(\Gm)^*\red B) \ \leftarrow\ \CSP(\core(\Gm)^*\red B\fac\th)
\]
\caption{Reductions between CSPs corresponding to derivative languages.}
\label{fig:reductions-o}
\end{figure}

\begin{theorem}\label{the:pp-o}
  Let $\Gm$ be a constraint language over $U_d$ and let $R$ be pp-definable from
  $\Gm$. Then, if $\CSP(\Gm\cup\{R\})$ has a satisfiability gap of the first (second, third) kind then so does
  $\CSP(\Gm)$.
\end{theorem}

Weaker forms of Theorem~\ref{the:pp-o} have appeared in the literature. This includes for example subdivisions from Theorem~6.4 of~\cite{Mastel24:stoc}.

Let $\rel\subseteq U_d^r$ be a pp-definable formula
over $\Gm$ via the pp-formula $\phi(x_1,\ldots,x_r)$ as in~(\ref{eq:pp-formula-o}).
Given an instance
$\cP\in\CSP(\Gm\cup\{\rel\})$, one can turn it into an instance
$\cP'\in\CSP(\Gm)$ that is equivalent to $\cP$. 
Intuitively, each constraint $\ang{\bu,\rel}$ of $\cP$ is replaced with
constraints 
from its pp-definition in~(\ref{eq:pp-formula-o})
over fresh new variables (and similarly for the binary equality relation, cf.~\Cref{sec:operator-pp} for details).
This construction is known as the \emph{gadget construction} in the CSP
literature and it is known that $\cP$ has a solution over $U_d$
if and only if $\cP'$ has a solution over $U_d$~\cite{Bulatov05:classifying,BKW17}.
Thus, to prove Theorem~\ref{the:pp-o}, it suffices to show the following;
the proof is a simple generalization of the $d=2$ case
proved in~\cite{AKS19:jcss}.

\begin{lemma}\label{lem:lift-o}
  Let $\Gm$ be a constraint language over $U_d$ and let $\rel$ be pp-definable
  from $\Gm$. Furthermore, let $\cP\in\CSP(\Gm\cup\{\rel\})$ and let
  $\cP'\in\CSP(\Gm)$ be the gadget construction replacing constraints involving
  $\rel$ in $\cP$. If there is a (finite or infinite dimensional) satisfying operator assignment for $\cP$ then there is a (respectively, finite or infinite dimensional) satisfying operator assignment for $\cP'$. 
\end{lemma}

If $\CSP(\Gm\cup\{\rel\})$ has a satisfiability gap of the first kind then there is
an unsatisfiable instance $\cP\in\CSP(\Gm\cup\{\rel\})$ with a satisfying
operator assignment. By~\cite{Bulatov05:classifying} (cf.
also~\cite{BKW17}), $\cP'$ is unsatisfiable. By~Lemma~\ref{lem:lift-o}, $\cP'$ has a satisfying operator assignment. 
Hence $\cP'$ establishes that $\CSP(\Gm)$ has a satisfiability gap, and
Theorem~\ref{the:pp-o} is proved. This argument also extends to gaps of the second and third kind.

\smallskip

Pp-definitions are the starting point of the algebraic approach to
CSPs~\cite{Bulatov05:classifying} and suffice for dealing with Boolean CSPs,
not only in~\cite{AKS19:jcss} but also in all papers on Boolean (variants of) CSPs. For
CSPs over larger domains, more tools are needed.

A constraint language $\Gm$ is a \emph{core} language if all its endomorphisms
are permutations; that is, $\Gm$ has no endomorphisms that are not
automorphisms. There always exists an endomorphism $\vr$ of $\Gm$ such that $\vr(\Gm)$ is
core and $\vr\circ\vr=\vr$~\cite{Bulatov05:classifying}. We will denote this core language by $\core(\Gm)$, as it (up to an isomorphism) does not depend on the choice of $\vr$.

A constraint language $\Gm$ is called \emph{idempotent} if it contains all the
\emph{constant} relations, that is, relations of the form $C_a=\{(a)\}$, $a\in
U_d$. For an arbitrary language $\Gm$ over $U_d$ we use $\Gm^*=\Gm\cup\{C_a\mid
a\in U_d\}$.  A unary relation (a set) $B\sse U_d$ pp-definable in $\Gm$ is
called a \emph{subalgebra} of $\Gm$. For a subalgebra $B$ we introduce the
\emph{restriction} $\Gm\red B$ of $\Gm$ to $B$ defined as $\Gm\red B=\{\rel\cap
B^{ar(\rel)}\mid \rel\in\Gm\}$.

An equivalence relation $\th$ pp-definable in $\Gm$ is said to be a \emph{congruence} of $\Gm$. The equivalence class of $\th$ containing $a\in U_d$ will be denoted by $a\fac\th$, and the set of all equivalence classes, the \emph{factor-set}, by $U_d\fac\th$. Congruences of a constraint language allow one to define a \emph{factor-language} as follows. For a congruence $\th$ of the language $\Gm$ the factor language $\Gm\fac\th$ is the language over $U_d\fac\th$ given by $\Gm\fac\th=\{\rel\fac\th\mid \rel\in\Gm\}$, where $\rel\fac\th=\{(a_1\fac\th\zd a_n\fac\th)\mid  (\vc an)\in\rel\}$.

All the languages above are related to each other 
by the reducibility of the corresponding CSPs, as Figure~\ref{fig:reductions-o}
indicates (cf. Proposition~\ref{pro:reductions}).
In~\Cref{sec:gap}, we show that all arrows in
Figure~\ref{fig:reductions-o} preserve satisfiability gaps.
To relate these reductions with bounded width and magic squares we use the
following result.

\begin{prop}[\cite{Bulatov09:width,Barto14:local,BKW17}]\label{pro:abelian-o}
For a constraint language $\Gm$ over $U_d$, $\CSP(\Gm)$ does not have bounded
  width if and only there exists a language $\Dl$ pp-definable in $\Gm$, a
  subalgebra $B$ of $\core(\Dl)^*$, a congruence $\th$ of $\core(\Dl)^*\red B$,
  and an Abelian group $\zA$ of prime order $p$ such that $\Gm'=\core(\Dl)^*\red B\fac\th$ contains relations $\rel_{3,a},\rel_{p+2}$ for every $a\in\zA$ given by $\rel_{3,a}=\{(x,y,z)\mid x+y+z=a\}$ and $\rel_{p+2}=\{(\vc a{p+2})\mid a_1+\dots+a_{p+2}=0\}$. If $\CSP(\Gm)$ is NP-hard, then $\Dl$ can be chosen to contain $\rel_{3,a},\rel_{p+2}$ for $p=2$.\footnote{The relations $\rel_{3,a},\rel_{p+2}$ are chosen here because they are needed for our purpose. In fact, they can be replaced with any relations expressible by linear equations over $\zA$.}
\end{prop}

To prove Theorem~\ref{the:hsp-gap-o}, suppose that $\CSP(\Gm)$ does not have
bounded width. 
By Proposition~\ref{pro:abelian-o} there is 
a language $\Dl$ pp-definable in $\Gm$, a subalgebra $B$ of $\core(\Dl)^*$, a congruence
$\th$ of $\core(\Dl)^*\red B$, and an Abelian group $\zA$ of prime order $p$
such that $\Gm'=\core(\Dl)^*\red B\fac\th$ contains relations $\rel_{3,a},\rel_{p+2}$
for every $a\in\zA$.
Our goal is to show that if $\CSP(\Gm')$ has a satisfiability gap
then so does $\CSP(\Gm)$.

To give an idea of the gap-preservation proofs, we sketch how a satisfiability
gap is preserved.
Let $\Gm$ be a constraint language over the set $U_d$ and let $B$ be its subalgebra.
We show that if $\CSP(\Gm\red B)$ has a satisfiability gap then so does
$\CSP(\Gm)$. We show this for a gap of the first kind (over finite-dimensional Hilbert spaces); the infinite-dimensional case is more complicated, cf.~\Cref{sec:gap} for details.
Let $\Dl=\Gm\red B$. Then by Theorem~\ref{the:pp-o} we may assume $\Dl\sse\Gm$ and $B\in\Gm$. Let $e=|B|$ and $\pi:U_e\to U_d$ a bijection between $U_e$ and $B$. 
Let $\cP=(V,U_e,\cC)$ be a gap instance of $\CSP(\pi^{-1}(\Dl))$ and the instance $\cP^\pi=(V,U_d,\cC^\pi)$ constructed as follows: For every $\ang{\bs,\rel}\in\cC$ the instance $\cP^\pi$ contains $\ang{\bs,\pi(\rel)}$. As is easily seen, $\cP^\pi$ has no classic solution. Therefore, it suffices to show that for any satisfying operator assignment $\{A_v\mid v\in V\}$ for $\cP$, the assignment $C_v=\pi(A_v)$ is a satisfying operator assignment for $\cP^\pi$.

By a technical lemma that shows that injective maps on finite sets that are interpolated
by polynomials preserve normal operators that pairwise commute (cf.
Lemma~\ref{lem:mapping}), 
the $C_v$'s are normal, satisfy the condition $C_v^d=I$, and locally commute. For $\ang{\bs,\rel}\in\cC$, $\bs=(\vc vk)$, let $f^\pi_\rel(\vc xk)=f_\rel(\pi^{-1}(x_1)\zd\pi^{-1}(x_k))$. 
It can be shown that $\pi^{-1}(C_v)=A_v$, and therefore $f^\pi_\rel(C_{v_1}\zd C_{v_k})=I$. For any $\vc ak\in U_d$, if $(\vc ak)\in\pi(\rel)$ then $\vc ak\in B$. Therefore, $f^\pi_\rel(\vc ak)=1$ then $f_{\pi(\rel)}(\vc ak)=1$. By Lemma~\ref{lem:lemma-3} this implies $f_{\pi(\rel)}(C_{v_1}\zd C_{v_k})=I$.

The other reductions use similar ideas, carefully relying on the spectral
theorem given in Theorem~\ref{the:SST} to simultaneously diagonalize the
restriction of an operator assignment to the scope of a constraint,
Lemma~\ref{lem:lemma-3} that relates polynomial equations over $U_d$ and
operators, and the above mentioned result on
preservation of operator assignments by certain polynomials (Lemma~\ref{lem:mapping}).
The infinite-dimensional case is more delicate, relying on the General Strong
Spectral Theorem (Theorem~\ref{the:general-sst}).

\smallskip
To finish the proof of Theorem~\ref{thm:main-informal}, assume that
$\CSP(\Gm)$ does not have bounded width and $\Gm'$ from
Proposition~\ref{pro:abelian-o} is over a group of prime order $p$. As reductions preserve satisfiability gaps, it suffices that
$\CSP(\Dl_p)$, where $\Dl_p=\{\rel_{3,a}\mid a\in\zA\}\cup\{\rel_{p+2}\}$, has a
satisfiability gap. The result of
Slofstra and Zhang~\cite{SZ24:arxiv} provides a gap instance of $\CSP(\{\rel_{3,1},\rel_{3,-1}\})$ of the 
second kind for any prime $p$, and thus the same holds for $\CSP(\Gm)$.
For $p=2$ the Mermin-Peres magic square from~\cite{Mermin1990simple}
provides a gap instance of $\CSP(\{\rel_{3,1},\rel_{3,-1}\})$ of the first kind,
and same holds for $\CSP(\Gm)$.
Finally, the result by Slofstra~\cite{Slofstra20:jams} provides a gap instance of
$\CSP(\{\rel_{3,1},\rel_{3,-1}\})$ of the third kind, and we get the same for
$\CSP(\Gm\cup\{\rel_T\})$, where $\rel_T$ is the full binary relation on $U_d$ (cf.
Theorem~\ref{the:pp} in~\Cref{sec:operator-pp} and the discussion
therein). Finally, if $\CSP(\Gm)$ is NP-hard then $\Gm$ can simulate any
constraint language~\cite{Bulatov05:classifying}, and thus in particular $\Delta_2$.

\section{Bounded width and no gaps}\label{sec:no-gap}

In this section we prove the first part of our main result.

\begin{theorem}\label{the:no-gap}
Let $\Gm$ be a constraint language over $U_d$. If $\CSP(\Gm)$ has bounded
  width then it has no satisfiability gap of any kind.
\end{theorem}

The main idea behind the proof
of Theorem~\ref{the:no-gap} is to simulate the inference provided by SLAC 
by inference in polynomial equations. Let $\cS$ be a SLAC-program solving $\CSP(\Gm)$.
In order to prove Theorem~\ref{the:no-gap} we take an instance $\cP=(V,U_d,\cC)$
of $\CSP(\Gm)$ that has no solution, and therefore is not SLAC-consistent, as
$\CSP(\Gm)$ has bounded width, and prove that it also has no satisfying operator
assignment. This rules out a gap of the second kind, and thus also a gap of any
kind (cf. the discussion in~\Cref{sec:operator-CSP}). We will prove it by contradiction, assuming $\cP$ has a satisfying operator assignment $\{A_v\}$ and then using the rules of a SLAC-program solving $\CSP(\Gm)$ to infer stronger and stronger conditions on $\{A_v\}$ that eventually lead to a contradiction. We start with a series of lemmas that will help to express the restrictions on $\{A_v\}$.

The following lemma introduces a restriction that is satisfied by any operator assignment. 

\begin{lemma}\label{lem:whole-domain-poly}
Let $\cP=(V,U_d,\cC)\in\CSP(\Gm)$. For any operator assignment $\{A_v\}$
for $\cP$ we have
\[
\prod_{k=0}^{d-1}(\ld_k I-A_v)=0.
\]
\end{lemma}
\begin{proof}
Note that the equation $\prod_{k=0}^{d-1}(\ld_k-x)=0$ is true for any
$x\in U_d$, that is, it follows from the empty set of equations. By
Lemma~\ref{lem:lemma-3}, it also holds for any fully commuting operator 
assignment. However, as the equation contains only one variable any 
operator assignment is fully commuting, and the result follows.
\end{proof}

Recall that every rule of a SLAC-program has the form $(x\in S)\meet\rel(x,y,\vc zr))\to (y\in S')$ for some variables $x,y\in V$, a constraint $\ang{(x,y,\vc zr),\rel}$, and sets $S,S'\sse U_d$. Therefore, we need to show how to encode unary relations and rules of a SLAC-program through polynomials. For any $S\sse U_d$, we represent the unary constraint restricting the domain of a variable
$x$ to the set $S$ by the polynomial
\[
Dom_S(x)=\prod_{k\in S}(\ld_k-x)+1.\footnote{This is not the representation of $S$ as in the beginning of~\Cref{sec:operator-CSP}, as $Dom_S(a)$ is not necessarily equal to $\ld_1$ for $a\not\in S$. However, it suffices for our purposes, because we only need the property that $Dom_S(a)=1$ if and only if $a\in S$.}
\]
Similarly, the rule $(x\in S)\meet\rel(x,y,\vc zr))\to (y\in S')$ of the SLAC 
program is represented by
\[
Rule_{S,\rel,S'}(x,y,\vc zr)=(Dom_{\ov S}(x)-1)(P_\rel(x,y,\vc zr)-\ld_1)
(Dom_{S'}(y)-1).
\]
As the next lemma shows, any operator assignment is a zero of $Rule_{S,\rel,S'}$.

\begin{lemma}\label{lem:rule-poly}
Let $\cP=(V,U_d,\cC)\in\CSP(\Gm)$. For any operator assignment $\{A_v\}$
for $\cP$ and any rule $(x\in S)\meet\rel(x,y,\vc zr))\to (y\in S')$ of the SLAC program for $\CSP(\Gm)$ we have
\begin{eqnarray*}
\lefteqn{Rule_{S,\rel,S'}(A_x,A_y,A_{z_1}\zd A_{z_r}) }\\
&=& (Dom_{\ov S}(A_x)-I)
(P_\rel(A_x,A_y,A_{z_1}\zd A_{z_r})-\ld_1I)(Dom_{S'}(A_y)-I)=0.
\end{eqnarray*}
\end{lemma}
\begin{proof}
Note that the equation $Rule_{S,\rel,S'}(x,y,\vc zr)=0$ is true for any
$x,y,\vc zr\in U_d$, that is, it follows from the empty set of equations. By
Lemma~\ref{lem:lemma-3}, it also holds for any fully commuting operator 
assignment. However, as all the variables $x,y,\vc zr$ belong to the scope
of the same constraint, the operators $A_x,A_y,A_{z_1}\zd A_{z_r}$ 
pairwise commute.  The result follows.
\end{proof}

Now, assume $\cP=(V,U_d,\cC)$ is not SLAC-consistent and $D_v$ denote the domain of $v\in V$ obtained after establishing SLAC-consistency. This means that for 
some $v\in V$ there is a derivation of $D_v=\eps$ using only facts
$\rel(\bs)$ for $\ang{\bs,\rel}\in\cC$ and 
$T_B(x_i)\meet\rel(\vc xk)\to T_C(x_j)$ for the rules of the SLAC-program. Moreover, this derivation can be subdivided into sections
of the form $(v=a)\to(v\ne a)$, each of which is linear. The latter 
condition means that each such section looks like a chain
$(v=a)\to(v_1\in S_1)\to\dots\to(v_\ell\in S_\ell)\to(v\in D_v-\{a\})$,
where each step is by a rule of the form 
$((v_i\in S_i)\meet\rel_i(v_i,v_{i+1},\vc ur))\to(v_{i+1}\in S_{i+1})$.

\begin{lemma}\label{lem:derivation-poly}
For any satisfying operator assignment $\{A_v\}$ for $\cP$ and any rule $(x\in S)\meet\rel(x,y,\vc zr))\to (y\in S')$ of the SLAC program for $\CSP(\Gm)$ if 
$\rel(x,y,\vc zr)\in\cC$ then 
\[
(Dom_{\ov S}(A_x)-I)(Dom_{S'}(A_y)-I)=I.
\]
\end{lemma}
\begin{proof}
By Lemma~\ref{lem:rule-poly}, the equation 
\[
(Dom_{\ov S}(A_x)-I)(P_\rel(A_x,A_y,A_{z_1}\zd A_{z_r})-\ld_1I)
(Dom_{S'}(A_y)-I)=0
\]
holds as well as the equation 
\[
P_\rel(A_x,A_y,A_{z_1}\zd A_{z_r})-I=0.
\]
Multiplying the latter one by $(Dom_{\ov S}(A_x)-I)$ on the left, and 
by $(Dom_{S'}(A_y)-I)$ on the right and subtracting it from the first equation
we obtain
\[
-(Dom_{\ov S}(A_x)-I)(1-\ld_1)I(Dom_{S'}(A_y)-I)=0.
\]
The result follows.
\end{proof}

\begin{lemma}\label{lem:transitive-poly}
Let $(v_1\in S_1)\to\dots\to(v_\ell\in S_\ell)$ be a derivation in the SLAC-program $\cS$ and $\{A_v\}$ a satisfying operator assignment for $\cP$. 
Then for each $i=2\zd\ell$
\[
(Dom_{\ov S_1}(A_{v_1})-I)(Dom_{S_i}(A_{v_i})-I)=0.
\]
\end{lemma}
\begin{proof}
We proceed by induction on $i$. For $i=2$ the equation holds by 
Lemma~\ref{lem:derivation-poly}.
In the inductive case we have equations
\begin{equation}\label{equ:equation1}
(Dom_{\ov S_1}(A_{v_1})-I)(Dom_{S_i}(A_{v_i})-I)=0,
\end{equation}
and 
\begin{equation}\label{equ:equation2}
(Dom_{\ov S_i}(A_{v_i})-I)(Dom_{S_{i+1}}(A_{v_{i+1}})-I)=0.
\end{equation}
The idea is to multiply (\ref{equ:equation1}) by 
$(Dom_{S_{i+1}}(A_{v_{i+1}})-I)$ on the right, multiply (\ref{equ:equation2})
by $(Dom_{\ov S_1}(A_{v_1})-I)$ on the left and subtract. The problem
is, however, that 
\[
Dom_{S_i}(A_{v_i})-Dom_{\ov S_i}(A_{v_i})
\]
is not a constant polynomial. So, we also need to prove that any polynomial
of the form 
\[
Dom_S(x)-Dom_{\ov S}(x)
\]
is invertible modulo $x^d-1$. The polynomial has the form
\[
p(x)=\prod_{k\in S}(x-\ld_k)-\prod_{k\not\in S}(x-\ld_k).
\]
As is easily seen, assuming that the product of an empty set of factors equals~1, $\ld_k$ is not a
root of $p(x)$ for any $\ld_k\in U_d$. Therefore the greatest common divisor of $p(x)$ and $x^d-1$ has degree 0, and hence there exists $q(x)$ such that 
\[
p(x)q(x)=c+r(x)(x^d-1).
\]

Thus before subtracting equations (\ref{equ:equation1}) and 
(\ref{equ:equation2}) we also multiply them by $q(A_{v_i})$. Then we get
\begin{eqnarray*}
(Dom_{\ov S_1}(A_{v_1})-I)(Dom_{S_i}(A_{v_i})q(A_{v_i})
-q(A_{v_i})Dom_{\ov S_i}(A_{v_i}))(Dom_{S_{i+1}}(A_{v_{i+1}})-I) &=& 0\\
(Dom_{\ov S_1}(A_{v_1})-I)q(A_{v_i})(Dom_{S_i}(A_{v_i})
-Dom_{\ov S_i}(A_{v_i}))(Dom_{S_{i+1}}(A_{v_{i+1}})-I) &=& 0\\
c(Dom_{\ov S_1}(A_{v_1})-I)(Dom_{S_{i+1}}(A_{v_{i+1}})-I) &=& 0.
\end{eqnarray*}
The first transformation uses the fact that $A_{v_i}$ commutes with itself, while the 
second one uses the property $A_{v_i}^d=I$.
The result follows.
\end{proof}

\begin{proof}[Proof of Theorem~\ref{the:no-gap}]
To complete the proof of Theorem~\ref{the:no-gap} note that the lack of
SLAC-consistency means that for some $v\in V$ the statement 
$(v=\ld_k)\to(v\ne\ld_k)$ can be derived from $\cP$ for every $\ld_k\in U_d$.
By Lemma~\ref{lem:transitive-poly}, for any operator assignment $\{A_w\}$
and any $\ld_k\in U_d$ the operator $A_v$ satisfies the equation
\[
\prod_{j\ne k}(A_v-\ld_jI)=0.
\]
We show that for any $S\sse U_d$ these equations imply 
\[
\prod_{j\in S}(A_v-\ld_jI)=0.
\]
Then for $S=\eps$ we get $I=0$, witnessing that $\cP$ has no satisfying 
operator assignment.

We proceed by (reverse) induction on the size of $S$. Suppose the statement is true 
for all sets of size $r$ and let $S\sse U_d$ be such that $|S|=r-1$. Without loss fo
generality, assume that
$S=\{\ld_0\zd\ld_{r-1}\}$. Let $S_1=S\cup\{\ld_r\}, S_2=S\cup\{\ld_{r+1}\}$.
Consider
\begin{equation}\label{equ:equation3}
\prod_{j\in S_1}(A_v-\ld_jI)=0
\end{equation}
and
\begin{equation}\label{equ:equation4}
\prod_{j\in S_2}(A_v-\ld_jI)=0.
\end{equation}
Subtracting (\ref{equ:equation4}) from (\ref{equ:equation3}) we obtain
\[
(A_v-\ld_rI-A_v+\ld_{r+1}I)\prod_{j=0}^{r-1}(A_v-\ld_jI)=
(\ld_{r+1}-\ld_r)I\prod_{j=0}^{r-1}(A_v-\ld_jI)=0,
\]
implying the equation for $S$.
\end{proof}

\section{Reductions through pp-definitions}\label{sec:operator-pp}

In this section we prove that the so-called primitive positive definitions, a key tool in the
algebraic approach to CSPs~\cite{Bulatov05:classifying}, not only give rise to
(polynomial-time) reductions that preserve satisfiability over $U_d$ but also
preserve satisfiability via operators. This was
established for the special case of Boolean domains (i.e., for $d=2$)
in~\cite{AKS19:jcss} and the same idea works for larger domains. We will need
this result later in~\Cref{sec:gap} to prove that certain CSPs of 
unbounded width admit a satisfiability gap.

Let $\Gm$ be a constraint language over $U_d$, let $r$ be an integer, and
let $x_1,\ldots,x_r$ be variables ranging over the domain $U_d$. A primitive
positive formula (\emph{pp-formula}) over $\Gm$ is a formula of the form
\begin{equation}\label{eq:pp-formula}
  \phi(x_1,\ldots,x_r)=\exists y_1\cdots\exists y_s(\rel_1(\bz_1)\wedge\cdots\wedge
  \rel_m(\bz_m)),
\end{equation}
where $\rel_i\in\Gm$ is either the binary equality relation on $U_d$ or a relation over $U_d$ of arity $r_i$ and each $\bz_i$ is an $r_i$-tuple of variables from
$\{x_1,\ldots,x_r\}\cup\{y_1,\ldots,y_s\}$.
A relation $\rel\subseteq U_d^r$ is primitive positive definable
(\emph{pp-definable}) from $\Gm$ if there
exists a pp-formula $\phi(x_1,\ldots,x_r)$ over $\Gm$ such that $\rel$ is equal to the set of
models of $\phi$, that is, the set of $r$-tuples $(a_1,\ldots,a_r)\in U_d^r$
that make the formula $\phi$ true over $U_d$ if $a_i$ is
substituted for $x_i$ in $\phi$ for every $i\in [r]$.

Let $\rel_T=U_d^2$ denote the full binary relation on $U_d$.
Our goal in this section is to prove the following result.

\begin{theorem}\label{the:pp}
  Let $\Gm$ be a constraint language over $U_d$ and let $R$ be pp-definable from
  $\Gm$.
  \begin{enumerate}
    \item If $\CSP(\Gm\cup\{R\})$ has a satisfiability gap of the first or second kind then so does $\CSP(\Gm)$.
    \item If $\CSP(\Gm\cup\{R\})$ has a satisfiability gap of the third kind then so does $\CSP(\Gm\cup\{R_T\})$.
  \end{enumerate}
\end{theorem}

Let $\rel\subseteq U_d^r$ be a pp-definable formula
over $\Gm$ via the pp-formula $\phi(x_1,\ldots,x_r)$ as
in~(\ref{eq:pp-formula}). Let $=_d$ denote the equality relation on $U_d$. 
Given an instance
$\cP\in\CSP(\Gm\cup\{\rel\})$ we describe a construction of an instance
$\cP'\in\CSP(\Gm\cup\{=_d\})$, and then transform $\cP'$ into an instance $\cP^=$ of $\CSP(\Gm)$ that is, in some sense, equivalent to $\cP$. 
We start with the instance $\cP$. For every constraint $\ang{\bu,\rel}$ of $\cP$ with
$\bu=(u_1,\ldots,u_r)$, we
introduce $s$ fresh new variables $t_1,\ldots,t_s$ for the quantified
variables in~(\ref{eq:pp-formula}); furthermore, we replace $\ang{\bu,\rel}$ by $m$ constraints
$\ang{\bw_i,\rel_i}$, $i\in [m]$, where $\bw_i$ is the tuple of variables obtained
from $\bz_i$ in~(\ref{eq:pp-formula}) by replacing $x_j$ by $u_j$, $j\in [r]$,
and by replacing $y_j$ by $t_j$, $j\in [s]$.
The collection of variables $u_1,\ldots,u_r,t_1,\ldots,t_s$ is called the
\emph{block} of the constraint $\ang{\bu,\rel}$ in $\cP'$.
This construction is known as the \emph{gadget construction} in the CSP
literature and it is known that $\cP$ has a solution over $U_d$
if and only if $\cP'$ has a solution over $U_d$~\cite{Bulatov05:classifying,BKW17}.

Next, suppose that $\cP'$ is an instance of $\CSP(\Gm\cup\{=_d\})$ with the set of variables $\{\vc xn\}$. The instance $\cP^=$ is obtained by identifying variables $x_i,x_j$ whenever there is the constraint $=_d(x_,x_j)$ in $\cP'$. More formally, let $\vr'$ denote the binary relation on $[n]$ given by $(i,j)\in\vr'$ if and only if $\cP'$ contains the constraint $=_d(x_i,x_j)$, and let $\vr$ be the symmetric-transitive closure of $\vr'$. Select a representative of every equivalence class of $\vr$; without loss of generality assume $[m]$ is the set of these representatives. Let $\vf:[n]\to[m]$ be the mapping that maps every $i\in[n]$ to the representative of its equivalence class. Then $\cP^=$ is the instance of $\CSP(\Gm)$ with the set of variables $\{\vc ym\}$ that, for every constraint $\ang{(x_{i_1}\zd x_{i_r}),\rel}$, where $\rel$ is not the equality relation, contains the constraint $\ang{(y_{\vf(i_1)}\zd y_{\vf(i_r)}),\rel}$.
Thus, in order to prove Theorem~\ref{the:pp}\,(1), it suffices to show the following lemma.

\begin{lemma}\label{lem:lift}
  Let $\Gm$ be a constraint language over $U_d$ and let $\rel$ be pp-definable
  from $\Gm$. Furthermore, let $\cP\in\CSP(\Gm\cup\{\rel\})$ and let
  $\cP'\in\CSP(\Gm\cup\{=_d\})$, $\cP^=\in\CSP(\Gm)$ be the gadget construction replacing constraints involving
  $\rel$ in $\cP$, and the instance obtained from $\cP'$ by identifying the variable related by the equality relation. If there is a satisfying operator assignment for $\cP$ then there is a satisfying operator assignment for $\cP^=$. 
\end{lemma}

Indeed, if $\CSP(\Gm\cup\{\rel\})$ has a satisfiability gap of the first or second kind then there is an unsatisfiable instance $\cP\in\CSP(\Gm\cup\{\rel\})$ that has a satisfying operator assignment. By the results in~\cite{Bulatov05:classifying} (cf.~also~\cite{BKW17}), $\cP',\cP^=$ are unsatisfiable. By~Lemma~\ref{lem:lift}, $\cP^=$ has a satisfying operator assignment. Hence $\cP^=$ establishes that $\CSP(\Gm)$ has a satisfiability gap of the same kind, as required to prove Theorem~\ref{the:pp}\,(1).

We will frequently use (below and also in~\Cref{sec:gap}) the following observations. We give a proof of them for the sake of completeness.

\begin{lemma}\label{lem:matrix-polys}
Let $f$ be a polynomial and $A,B$ operators on a Hilbert space.\\[2mm]
(1) If $A,B$ commute, then so do $f(A),f(B)$.\\
(2) If $A$ is normal, then so is $f(A)$.\\
(3) If $U$ is a unitary operator, then $Uf(A)U^{-1}=f(UAU^{-1})$.\\[1mm]
(4) Let $g(x_1,\dots,x_r)$ be a multi-variate polynomial and $A_1,\dots,A_r$
  operators on a Hilbert space. If $U$ is a unitary operator, then 
\[
Ug(A_1,\dots,A_r)U^{-1}=g(UA_1U^{-1},\dots,UA_rU^{-1}).
\]
\end{lemma}

\begin{proof}
Let $f(x)=\sum_{i=0}^k\al_ix^i$.\\[2mm]
(1) We have
\begin{align*}
f(A)f(B) &=\left(\sum_{i=0}^k\al_iA^i\right)\left(\sum_{i=0}^k\al_iB^i\right)=\sum_{i,j=0}^k\al_i\al_jA^iB^j\\
&=\sum_{i,j=0}^k\al_j\al_iB^jA^i=\left(\sum_{i=0}^k\al_iB^i\right)\left(\sum_{i=0}^k\al_iA^i\right)=f(B)f(A).
\end{align*}

\noindent
(2) The operator $A$ is normal if it commutes with its 
adjoint $A^*$. Since for any operators $B,C$ it holds that $(B+C)^*=B^*+C^*$ and $(B^k)^*=(B^*)^k$, we obtain $f(A^*)=f(A)^*$. The result then follows from item (1).

\noindent
(3) We have
\begin{align*}
Uf(A)U^{-1} &=U\left(\sum_{i=0}^k\al_iA^i\right)U^{-1} =\sum_{i=0}^kU\al_iA^iU^{-1} =\sum_{i=0}^k\al_i(UAU^{-1})(UAU^{-1})\dots(UAU^{-1})\\ 
&=\sum_{i=0}^k\al_i(UAU^{-1})^i=f(UAU^{-1}).
\end{align*}

\noindent
(4) In this case the proof is similar to that of (3). Let
\[
g(x_1,\dots,x_r)=\sum_{i=(i_1,\dots,i_r)\in\cal I}\al_ix_1^{i_1}\dots x_r^{i_r},
\]
where $\cI$ is a finite subset of $(\zN\cup\{0\})^r$. We have
\begin{align*}
Ug(A_1,\dots,A_r)U^{-1} &=U\left(\sum_{i=(i_1,\dots,i_r)\in\cal I}\al_iA_1^{i_1}\dots A_r^{i_r}\right)U^{-1}\\
&=\sum_{i=(i_1,\dots,i_r)\in\cal I}U\al_iA_1^{i_1}\dots A_r^{i_r}U^{-1}\\ 
&=\sum_{i=(i_1,\dots,i_r)\in\cal I}\al_i(UA_1^{i_1}U^{-1})\dots(UA_r^{i_r}U^{-1})\\ 
&=\sum_{i=(i_1,\dots,i_r)\in\cal I}\al_i(UA_1U^{-1})^{i_1}\dots(UA_rU^{-1})^{i_r}\\
&=g(UA_1U^{-1},\dots,UA_rU^{-1}).
\end{align*}
\end{proof}

\begin{proof}[Proof of Lemma~\ref{lem:lift}] 
We prove the lemma in two steps. First we show that $\cP$ has a satisfying assignment if and only $\cP'$ does and then prove the same connection for $\cP'$ and $\cP^=$. 

{\it Finite-dimensional case.}
Let $\cP=(V,U_d,\cC)$ and let $\{A_v\}_{v\in V}$ be an
  operator assignment that is satisfying for $\cP$ on a finite-dimensional
  Hilbert space $\cH$. We may assume that $\cH=\zC^p$ for some positive integer
  $p$
  and thus $\{A_v\}_{v\in V}$ are $p\times p$ matrices. We will
  construct an operator assignment that is satisfying for $\cP'$; it will be an
  operator assignment on the same space and an extension of the original
  assignment.
 
  Given a constraint $\ang{(u_1,\ldots,u_r),\rel}\in\cC$, the operators
  $\{A_{u_i}\}_{1\leq i\leq r}$ pairwise commute by assumption. Since they are
  also normal, by Theorem~\ref{the:SST} there is a unitary matrix $U$ such that
  $E_i=UA_{u_i}U^{-1}$ is a diagonal matrix for each $i\in [r]$. Since $A_{u_i}^d=I$,
  we have $E_{u_i}^d=I$. Thus every diagonal entry $E_{u_i}(jj)$ belongs to
  $U_d$. Since $P_\rel(A_{u_1},\ldots,A_{u_r})=I$, we have $P_\rel(E_{u_1},\ldots,E_{u_r})=UP_\rel(A_{u_1},\ldots,A_{u_r})U^{-1}=I$ by Lemma~\ref{lem:matrix-polys}. As
  $E_{u_i}$ is a diagonal matrix, we have that
  \[P_\rel(E_{u_1}(jj),\ldots,E_{u_r}(jj))=1=\lambda_0\]
  for every $j\in [p]$.
  Thus the tuple $(E_{u_1}(jj),\ldots,E_{u_r}(jj))\in\rel$ for every $j\in [p]$. 
  For each $j\in [p]$, let $b(j)=(b_1(j),\ldots,b_s(j))\in U_d^s$ be a tuple of
  witnesses to the existentially quantified variables $y_1,\ldots,y_s$ in the
  formula pp-defining $\rel$ in~(\ref{eq:pp-formula}) when $E_{u_i}(jj)$ is
  substituted for $x_i$. For $i\in [s]$, let $B_i$ be the diagonal matrix with
  $B_i(jj)=b_i(j)$ for every $j\in [d]$. Now let $C_i=U^{-1}B_iU$. Since $U$ is
  unitary we have $U^{-1}=U^*$ and thus each $C_i$ is 
  normal:
  $C_iC_i^*=U^*B_iU(U^*B_iU)^*=U^*B_iUU^*B_i^*U
  =U^*B_iB_i^*U
  =U^*B_i^*B_iU
  =U^*B_i^*UU^*B_iU
  =(U^*B_iU)^*U^*B_iU=C_i^*C_i$.
  Since $b_i(j)\in U_d$, we have $C_i^d=I$. Also
  $E_1,\ldots,E_r,B_1,\ldots,B_s$ pairwise commute since they are all diagonal
  matrices. Hence, $A_i,\ldots,A_r,C_1,\ldots,C_s$ also pairwise commute since
  they all are simultaneously similar via $U$.
  As each conjunct in~(\ref{eq:pp-formula}) is satisfied by the assignment
  sending $x_i$ to $E_{u_i}(jj)$ and $y_i$ to $b_i(j)$ for all $j\in [d]$, we
  can conclude that the matrices that are assigned to the variables in the
  conjuncts make the corresponding polynomial evaluate to $I$. But
  this means that the assignment to the variables in the block of the constraint
  $\ang{(u_1,\ldots,u_r),\rel}$ makes a satisfying operator assignment for the
  constraint of $\cP'$ that has come from the conjunct. As different constraints involving $\rel$ in $\cP$
  produce their own sets of fresh variables, the operator assignments do not
  affect each other.

\smallskip

{\it Infinite-dimensional case.}
Let $\cP=(V,U_d,\cC)$ and let $\{A_v\}_{v\in V}$ be a collection of bounded
  normal linear operators on a Hilbert space $\cH$ that is satisfying for $\cP$. We need to define bounded
normal linear operators for the new variables of $\cP'$ that were introduced in
  the gadget construction. We define these operators
simultaneously for all variables $t_1,\dots,t_s$ that come from the same constraint $\ang{(u_1,\ldots,u_r),\rel}\in\cC$ of $\cP$.

By the commutativity condition of satisfying operator assignments, the operators
  $A_1,\dots,A_r$ pairwise commute. Each $A_i$ is normal, and so the General Strong Spectral Theorem (cf. Theorem~\ref{the:general-sst}) applies.
Thus, there exist a measure space $(\Omega,\cM,\mu)$, a unitary map $U:\cH\to L^2(\Omega,\mu)$, and functions $a_1,\dots,a_r\in L^\infty(\Omega,\mu)$ such
that, for the multiplication operators $E_i = T_{a_i}$ of $L^2(\Omega,\mu)$, the relations $A_i = U^{-1} E_iU$ hold for each $i\in [r]$. Equivalently,
$U A_iU^{-1} = E_i$. From $A_i^d=I$ we conclude $E_i^d=I$. Hence, $a_i(\omega)^d = 1$ for almost all $\omega\in\Omega$; i.e., formally $\mu(\{\omega\in\Omega\mid a_i(\omega)^d\ne1\}) = 0$. Thus, $a_i(\omega)\in U_d$ for almost all $\omega\in\Omega$. The conditions of Lemma~\ref{lem:matrix-polys} apply, thus $P_R(E_1,\dots, E_r)=U^{-1}P_R (A_1,\dots, A_r)U=1$. Now, $E_i$ is the multiplication operator given by $a_i$, and $a_i(\omega)\in U_d$ for almost all $\omega\in\Omega$, so $P_R (a_1(\omega),\dots,a_r(\omega)) = 1$ for almost all $\omega\in\Omega$. Thus the tuple $a(\omega) = (a_1(\omega),\dots,a_r(\omega))$ belongs to the relation
$R$ for almost all $\omega\in\Omega$. Now we are ready to define the operators for the variables $t_1,\dots,t_s$.

For each $\omega\in\Omega$ for which the tuple $a(\omega)$ belongs to $R$, let $b(\omega) = (b_1(\omega),\dots, b_s(\omega))\in U_d^s$ be the lexicographically
smallest tuple of witnesses to the existentially quantified variables in $\phi(a_1(\omega),\dots, a_r(\omega))$; such a vector of witnesses
must exist since $\phi$ defines $R$, and the lexicographically smallest exists because $R$ is finite. For every other $\omega\in\Omega$, define
$b(\omega) = (b_1(\omega),\dots, b_s(\omega)) = (0,..., 0)$.

Note that each function $b_k:\Omega\to\zC$ is bounded since its range is in $U_d\cup\{0\}$. We claim that such functions of witnesses
$b_k$ are also measurable functions of $(\Omega,\cM,\mu)$. This will follow from the fact that $a_1,\dots,a_r$ are measurable functions themselves, the fact that $R$ is a finite relation, and the choice of a definite tuple of witnesses of each $\omega\in\Omega$: the lexicographically
smallest if $a(\omega)$ is in $R$, or the all-zero tuple otherwise. 

Since $R$ is finite, the event $Q = \{\omega\in\Omega\mid b_k(\omega) =\sigma\}$, for fixed $\sigma\in U_d\cup\{0\}$, can be expressed as a finite Boolean
combination of events of the form $Q_{i,\tau} = \{\omega\in\Omega\mid a_i(\omega) = \tau\}$, where $i\in[r]$ and $\tau\in U_d$. Here is how: If $\sigma\ne0$, then
\[
Q=\bigcup_{a\in R:b(a)_k=\sigma}\left(\bigcap_{i\in r} Q_{i,a_i}\right),
\]
where $b(a)$ denotes the lexicographically smallest tuple of witnesses in $U_d^s$ for the quantified variables in $\phi(a_1,\dots,a_r)$. If $\sigma = 0$, then $Q$ is the complement of this set. Each $Q_{i,\tau}$ is a measurable set in the measure space $(\Omega,\cM,\mu)$ since
$a_i$ is a measurable function and $Q_{i,\tau} = a^{-1}_i(B_{1/2d}(\tau))$, where $B_{1/2d}(\tau)$ denotes the complex open ball of radius $1/2d$ centered
at $\tau$, which is a Borel set in the standard topology of $\zC$. Since the range of $b_k$ is in the finite set $U_d\cup\{0\}$, the preimage
$b^{-1}_k(S)$ of each Borel subset $S$ of $\zC$ is expressed as a finite Boolean combination of measurable sets, and is thus measurable
in $(\Omega,\cM,\mu)$.

We just proved that each $b_k$ is bounded and measurable, so its equivalence class under almost everywhere equality
is represented in $L^\infty(\Omega,\mu)$. We may assume without loss of generality that $b_k$ is its own representative; else modify it
on a set of measure zero in order to achieve so. Let $F_k = T_{b_k}$ be the multiplication operator given by $b_k$ and let $t_k$ be
assigned the linear operator $B_k = U^{-1} F_kU$, which is bounded because $b_k$ is bounded and $U$ is unitary. Also because $U$
is unitary, each such operator is normal and $B_k^d=I$ since $b_k(\omega)\in U_d$ for almost all $\omega\in\Omega$. Moreover, $E_1,\dots, E_r, F_1,\dots, F_s$ pairwise commute since they are multiplication operators; thus $A_1,\dots, A_r, B_1,\dots, B_s$ pairwise
commute since they are simultaneously similar via $U$. Moreover, as each atomic formula in the quantifier-free part of $\phi$
is satisfied by the mapping that sends $u_i\mapsto a_i(\omega)$ and $t_i\mapsto b_i(\omega)$ for almost all $\omega\in\Omega$, another application of Lemma~\ref{lem:matrix-polys}
shows that the operators that are assigned to the variables of this atomic formula make the corresponding indicator polynomial evaluate to $I$. This means that the assignment to the variables $u_1,\dots, u_r,t_1,\dots,t_s$ makes a satisfying operator assignment for
the constraints of $\cP'$ that come from the constraint $\ang{(u_1,\dots, u_r), R}$ in $\cP$. As different constraints from $\cP$ produce their own sets of
variables $t_1,\dots,t_s$, these definitions of assignments are not in conflict with each other, and the proof of the lemma is complete.

\smallskip

Let $P_=$ be the polynomial interpolating the equality relation, and let $W$ be the set of variables of $\cP'$. As is easily seen, to prove the second step it suffices to prove that for any satisfying operator assignment $\{A_w\}_{w\in W}$ for $\cP'$ if $\ang{(u,v),=_d}$ is a constraint in $\cP'$ then $A_u=A_v$. In the finite-dimensional case let $\{A_w\}_{w\in W}$ be an operator assignment over a $p$-dimensional Hilbert space, $\ang{(u,v),=_d}$ a constraint in $\cP'$, and let $E_u=UA_uU^{-1}, E_v=UA_vU^{-1}$ for some unitary operator $U$ be diagonal operators. Then as before $P_=(E_u(jj),E_v(jj))=1$ for all $j\in[p]$ implying $E_u=E_v$, and hence $A_u=A_v$.

The infinite-dimensional case is similar, except we use the General Strong Spectral Theorem to find a measure space $(\Omega,\cM,\mu)$, a unitary map $U:\cH\to L^2(\Omega,\mu)$, and functions $a,b\in L^\infty(\Omega,\mu)$ such
that, for the multiplication operators $E_u = T_a, E_v=T_b$ of $L^2(\Omega,\mu)$, the relations 
$U A_uU^{-1} = E_u, UA_vU^{-1}=E_v$ hold. Since $P_=(A_u,A_v)=I$, we conclude that $a(\omega)=b(\omega)$ for almost all $\omega\in\Omega$ implying $E_u=E_v$, and therefore $A_u=A_v$. 
\end{proof}

We do not know whether the converse of Lemma~\ref{lem:lift} holds; this is not known even in the case of $d=2$~\cite{AKS19:jcss}.
The obvious idea would be to take the restriction of the operator assignment
that is satisfying for $\cP'$ but it is not clear why this should be satisfying
for $\cP$, because there is no guarantee that the operators assigned to variable in the scope of a constraint of $\cP$ of the form $\ang{\bs, \rel}$ commute. However, under a slight technical assumption on $\Gm$ --- namely,
that it includes the full binary relation $\rel_T$ on $U_d$\footnote{This is a special case of the
so-called \emph{commutativity gadget}~\cite{AKS19:jcss}; the ``T'' stands for
trivial.\\ Also, although the full binary relation can be pp-defined from the equality relation, it does not help to prove Theorem~\ref{the:pp}\,(2) as the full binary relation is needed exactly to deal with pp-definitions.} --- one can
enforce commutativity within a constraint scope and thus project an operator
assignment. 

For an instance $\cP'$ as defined above (and in the statement of Lemma~\ref{lem:lift}), we denote by $\cP''$ the instance obtained from $\cP'$ by adding, for every constraint $\ang{(\vc ur),\rel}$ of $\cP$, constraints of the form $\ang{(u_i,u_j),\rel_T}$ for every $i\neq j\in [r]$.
To prove Theorem~\ref{the:pp}\,(2), it suffices to show the following lemma.

\begin{lemma}\label{lem:proj}
  Let $\Gm$ be a constraint language over $U_d$ with $\rel_T\in\Gm$ and let $\rel$ be pp-definable
  from $\Gm$. Furthermore, let $\cP\in\CSP(\Gm\cup\{\rel\})$ and let
  $\cP''\in\CSP(\Gm)$ be defined as above.  Then, we have the following:\\[2mm]
  (1) For every satisfying operator assignemnt for $\cP$ on a Hilbert
  space $\cH$ there is an extension that is a satisfying operator
  assignment for $\cP''$ on $\cH$.\\[2mm]
  (2) For every satisfying operator assignment for $\cP''$ on $\cH$, the
  restriction of it onto the variables of $\cP$ is a satisfying operator
  assignment for $\cP$ on $\cH$.
\end{lemma}
Indeed, if $\CSP(\Gm\cup\{\rel\})$ has a satisfiability gap of the
third kind then there is an instance $\cP\in\CSP(\Gm\cup\{\rel\})$ that is not
satisfiable via finite-dimensional operators but has a satisfying
infinite-dimensional operator assignment. By~Lemma~\ref{lem:proj}, $\cP''$ is
not satisfiable via finite-dimensional operators but has a satisfying
infinite-dimensional operator assignment. Hence $\cP''$ establishes that
$\CSP(\Gm\cup\{\rel_T\})$ has a satisfiability gap of the third kind, as required to prove Theorem~\ref{the:pp}\,(2).
\begin{proof}[Proof of Lemma~\ref{lem:proj}]
(1) follows from Lemma~\ref{lem:lift}: The satisfying operator assignment for $\cP'$ constructed in the proof of Lemma~\ref{lem:lift} is also a satisfying operator assignment or $\cP''$. Indeed, the constraints already present in $\cP'$ are by assumption satisfied in $\cP''$. Regarding the extra constraints in $\cP''$ not present in $\cP'$, each such constraint involves the $\rel_T$ relation and the two variables in the scope of the constraint come from the block of some constraint in $\cP'$. The proof of Lemma~\ref{lem:lift} established that the constructed operators pairwise commute on these variables. Also, as $P_{\rel_T}(a,b)=1$ for any $a,b\in U_d$, the polynomial constraints are satisfied as they evaluate to $I$.

For~(2), take a satisfying operator assignment for $\cP''$ and consider its restriction $\{A_v\}_{v\in V}$ onto the variables of $\cP$. Any two operators whose variables appear within the scope of some constraint of $\cP$ necessarily commute since the two variables are in the block of some constraint in $\cP''$. It remains to show that the polynomial constraints of $\cP$ are satisfied, that is, that $P_\rel(A_{u_1},\ldots,A_{u_r})=I$ for every constraint $\ang{(u_1,\ldots,u_r),\rel}$ of $\cP$.
For this, we use Lemma~\ref{lem:lemma-3}.
Let $\phi$ be the formula pp-defining $\rel$ as in~(\ref{eq:pp-formula}). We
  define several polynomials over variables $x_1,\ldots,x_r,y_1,\ldots,y_s$ that
  correspond to the variables in~(\ref{eq:pp-formula}). For every $i\in [m]$,
  let $Q_i$ be the polynomial $P_{\rel_i}(\bz_i)-1$ so that the equation $Q_i=0$
  ensures $P_{\rel_i}(\bz_i)=1$, where $P_{\rel_i}$ is the characteristic
  polynomial of $\rel_i$, and $\bz_i$ is the tuple of variables from
  $x_1,\ldots,x_r,y_1,\ldots,y_s$ that correspond to the variables of the same
  name that appear in the conjunct $\rel_i(\bz_i)$ of~(\ref{eq:pp-formula}).
Let $Q$ be the polynomial $P_\rel(x_1,\ldots,x_r)-1$, where $P_\rel$ is the characteristic polynomial of $\rel$. By the construction of the polynomials and the choice of $\phi$, every assignment over $U_d$ that satisfies all equations $Q_1=\cdots=Q_m=0$ also satisfies $Q=0$. By Lemma~\ref{lem:lemma-3}, we get $P_\rel(A_{u_1},\ldots,A_{u_r})-I=0$, as required.
\end{proof}

\section{Unbounded width and gaps}\label{sec:gap}

In this section we prove the second part of our main result: we show that one
can ``implement'' linear equations over an Abelian group of prime order $p$
in $\CSP(\Gm)$ provided that $\CSP(\Gm)$ does not have bounded width. The rest
of Theorem~\ref{thm:main-informal}\,(2b) will then follow from the existence of a
satisfiability gap of the second kind~\cite{SZ24:arxiv} and, for $p=2$, the
existence of a satisfiability gap of the first
kind~\cite{Mermin1990simple,Mermin1993hidden,Peres1990incompatible} and the
third kind~\cite{Slofstra20:jams}, as all these gap instances are just systems
of linear equations. Moreover, for NP-hard CSPs one can ``implement'' linear
equations over $\mathbb{Z}_2$, and thus obtains satisfiabiilty gaps of all three
kinds, as claimed in Theorem~\ref{thm:main-informal}\,(1).

We achieve our goal in several steps via a chain of reductions that has been
used since the inception of the algebraic method to the
CSP~\cite{Bulatov05:classifying}. While more direct constructions have been
developed later, see, e.g., \cite{BKW17}, we find this original approach to be
better suited for operator CSPs.

\subsection{Bounded width and Abelian groups}

We start by introducing several definitions. 

A constraint language $\Gm$ over $U_d$ is said to be a \emph{core language} if
its every endomorphism is a permutation.
This term comes from finite model theory where it is used for relational structures that do not have endomorphisms (homomorphisms to themselves) that are not automorphisms. Such structures, and therefore languages, have a number of useful properties that we will exploit later. The standard way to convert a constraint language $\Gm$ into a core language is to repeat the following procedure until the resulting language is a core language: Pick an endomorphism $\vr$ of $\Gm$ that is not a permutation and set 
\[
\vr(\Gm)=\{\vr(\rel)\mid \rel\in\Gm\}, \quad\text{where } \vr(\rel)=\{(\vr(a_1)\zd \vr(a_n))\mid  (\vc an)\in\rel\}.
\]
There always exists an endomorphism $\vr$ of $\Gm$ such that $\vr(\Gm)$ is
core~\cite{Bulatov05:classifying,BKW17} and $\vr$ is \emph{idempotent}, that is, $\vr\circ\vr=\vr$. We will denote this core language by $\core(\Gm)$, as it (up to an isomorphism) does not depend on the choice of $\vr$. Note that the fact that $\vr$ is idempotent implies that it acts as identity on its image.

The language $\Gm$ is called \emph{idempotent} if it contains all the \emph{constant} relations, that is, relations of the form $C_a=\{(a)\}$, $a\in U_d$. For an arbitrary language $\Gm$ over $U_d$ we use $\Gm^*=\Gm\cup\{C_a\mid a\in U_d\}$.  A unary relation (a set) $B\sse U_d$ pp-definable in $\Gm$ is called a \emph{subalgebra} of $\Gm$. For a subalgebra $B$ we introduce the \emph{restriction} $\Gm\red B$ of $\Gm$ to $B$ defined as follows
 \[
\Gm\red B=\{\rel\cap B^{ar(\rel)}\mid \rel\in\Gm\}.
\]

An equivalence relation $\th$ pp-definable in $\Gm$ is said to be a \emph{congruence} of $\Gm$. The equivalence class of $\th$ containing $a\in U_d$ will be denoted by $a\fac\th$, and the set of all equivalence classes, the \emph{factor-set}, by $U_d\fac\th$. Congruences of a constraint language allow one to define a \emph{factor-language} as follows. For a congruence $\th$ of the language $\Gm$ the factor language $\Gm\fac\th$ is the language over $U_d\fac\th$ given by 
\[
\Gm\fac\th=\{\rel\fac\th\mid \rel\in\Gm\}, \quad\text{where } \rel\fac\th=\{(a_1\fac\th\zd a_n\fac\th)\mid  (\vc an)\in\rel\}.
\]

In order to fit core languages, subalgebras, and factor-languages in our framework where the domain is the set of roots of unity, we let $e=|\vr(U_d)|$, $e=|B|$ or $e=|U_d\fac\th|$, respectively, arbitrarily choose a bijection $\pi:\vr(U_d)\to U_e$, $\pi:B\to U_e$, and $\pi:U_d\fac\th\to U_e$, and replace $\vr(\Gm)$, $\Gm\red B$, and $\Gm\fac\th$ with $\pi(\vr(\Gm))$, $\pi(\Gm\red B)$, and $\pi(\Gm\fac\th)$, respectively. 

All the languages above are connected with each other in terms of the reducibility of the corresponding CSPs, as Figure~\ref{fig:reductions} and the following statements indicate.

\begin{prop}[\cite{Bulatov05:classifying,BKW17}]\label{pro:reductions}
Let $\Gm$ be a constraint language over $U_d$. Then
\begin{itemize}
\item[(1)]
  $\CSP(\Gm)$ and $\CSP(\core(\Gm))$ are polynomial-time interreducible.
\item[(2)]
If $\Gm$ is a core language, $\CSP(\Gm)$ and $\CSP(\Gm^*)$ are polynomial-time interreducible.
\item[(3)]
If $B$ is a subalgebra of $\Gm$ then $\CSP(\Gm\red B)$ is polynomial-time reducible to $\CSP(\Gm)$.
\item[(4)]
If $\th$ is a congruence of $\Gm$ then $\CSP(\Gm\fac\th)$ is polynomial-time reducible to $\CSP(\Gm)$.
\end{itemize}
\end{prop}

\begin{figure}
\[
\CSP(\Gm)\ \leftrightarrow\ \CSP(\core(\Gm)) \leftrightarrow\ \CSP(\core(\Gm)^*)\ \leftarrow\ \CSP(\core(\Gm)^*\red B) \ \leftarrow\ \CSP(\core(\Gm)^*\red B\fac\th)
\]
\caption{Reductions between CSPs corresponding to derivative languages}\label{fig:reductions}
\end{figure}

Finally, to relate the reductions above with bounded width we apply the
following result that can be extracted from the known results on the algebraic
approach to CSPs~\cite{Bulatov09:width,Barto14:local,BKW17}. Also, the dichotomy
theorem of Bulatov and Zhuk~\cite{Bulatov17:focs,Zhuk20:jacm} implies that
any NP-hard CSP can implement linear equations over $\mathbb{Z}_2$.

\begin{prop}[\cite{Bulatov09:width,Barto14:local,BKW17}]\label{pro:abelian}
For a constraint language $\Gm$ over $U_d$, $\CSP(\Gm)$ does not have bounded
  width if and only there exists a language $\Dl$ pp-definable in $\Gm$, a
  subalgebra $B$ of $\core(\Dl)^*$, a congruence $\th$ of $\core(\Dl)^*\red B$,
  and an Abelian group $\zA$ of prime order $p$ such that $\Gamma'=\core(\Dl)^*\red B\fac\th$ contains relations $\rel_{3,a},\rel_{p+2}$ for every $a\in\zA$ given by 
\[
\rel_{3,a}=\{(x,y,z)\mid x+y+z=a\},\quad\text{and}\quad \rel_{p+2}=\{(\vc a{p+2})\mid a_1+\dots+a_{p+2}=0\}.\footnote{The relations $\rel_{3,a},\rel_{p+2}$ are chosen here because they are needed for our purpose. In fact, they can be replaced with any relations expressible by linear equations over $\zA$.}
\]
  Moreover, if $\CSP(\Gm)$ is NP-hard then one can take $p=2$.
\end{prop}

We now have everything to formally state our second main result.

\begin{theorem}\label{the:hsp-gap}
  Let $\Gm$ be a constraint language over $U_d$ such that $\CSP(\Gm)$ does not
  have bounded width. Furthermore, let $\Gm'$ be the language guaranteed by
  Proposition~\ref{pro:abelian}. Then, if $\CSP(\Gm')$ has a satisfiability gap
  of the first or the second kind then so does $\CSP(\Gm)$.  Finally, if
  $\CSP(\Gm')$ has a satisfiability gap of the third kind then so does
  $\CSP(\Gm\cup\{\rel_T\})$.
\end{theorem}

\subsection{Proof of Theorem~\ref{the:hsp-gap}}

In this section we prove that the connections shown in
Figure~\ref{fig:reductions} hold in terms of satisfiability gaps, thus proving
Theorem~\ref{the:hsp-gap}. We start with a helpful observation. As any mapping on a finite set of complex numbers can be interpolated by a polynomial, we may apply such mappings to operators as well (we assume that such an interpolating polynomial is of the lowest degree possible, and so is unique). A polynomial $\vr$ is said to \emph{interpolate} a set $B\sse U_d$ if $\vr(\ld)=1$ if $\ld\in B$ and $\vr(\ld)=0$ if $\ld\in U_d-B$.

\begin{lemma}[Finite-dimensional case]\label{lem:mapping}
Let $d,e\in\nat$, $e\le d$. 
Let $\pi:U_e\to U_d$ be an injective mapping and $\vr$ a unary polynomial that interpolates $B=\Im(\pi)$.  Let $\vc Ak$ be pairwise commuting normal operators of order $e$ on a finite dimensional Hilbert space. Then $C_i=\pi(A_i)$, $i\in[k]$, are pairwise commuting normal operators of order $d$, and $\vr(C_i)=I$. Conversely, let $\vc Ck$ be pairwise commuting normal operators of order $d$ such that $\vr(C_i)=I$. Then for $A_i=\pi^{-1}(C_i)$, $i\in[k]$, it holds that the $A_i$'s are pairwise commuting normal operators of order $e$, and the eigenvalues of the $C_i$'s belong to $B$.
\end{lemma}

\begin{proof}
That the $C_i$'s are normal and pairwise commute follow from Lemma~\ref{lem:matrix-polys}. Let the $A_i$'s be $\ell$-dimensional and $U$ a unitary operator guaranteed by Theorem~\ref{the:SST} such that $UA_iU^{-1}$ is diagonal for all $i\in[k]$, and let 
\[
UA_iU^{-1}=\left(\begin{array}{ccc}\mu_{i1}&\dots&0\\ \vdots&\ddots&\vdots\\ 0&\dots&\mu_{i\ell} \end{array}\right).
\]
Then,
\begin{eqnarray*}
C_i^d &=& (\pi(A_i))^d=U^{-1}U(\pi(A_i))^dU^{-1}U=U^{-1}(\pi(UA_iU^{-1}))^dU\\
&=& U^{-1}\left(\begin{array}{ccc}(\pi(\mu_{i1}))^d&\dots&0\\ \vdots&\ddots&\vdots\\ 0&\dots&(\pi(\mu_{i\ell}))^d \end{array}\right)U = U^{-1}\left(\begin{array}{ccc}1&\dots&0\\ \vdots&\ddots&\vdots\\ 0&\dots&1 \end{array}\right)U=I,
\end{eqnarray*}
 because $\pi(\mu_{ij})\in U_d$ and by Lemma~\ref{lem:matrix-polys}. In a similar way, as $\pi(\mu_{ij})\in B$,
 \[
 \vr(C_i)= U^{-1}\left(\begin{array}{ccc}\vr(\mu_{i1})&\dots&0\\ \vdots&\ddots&\vdots\\ 0&\dots&\vr(\mu_{i\ell}) \end{array}\right)U=I.
 \]

For the second part of the claim we first need verify that all the eigenvalues of the $C_i$'s belong to $B$. Let $U$ be a unitary operator that diagonalizes the $C_i$'s and 
\[
UC_iU^{-1}=\left(\begin{array}{ccc}\mu_{i1}&\dots&0\\ \vdots&\ddots&\vdots\\ 0&\dots&\mu_{i\ell} \end{array}\right),
\]
assuming that the $C_i$'s are $\ell$-dimensional.  Then we have
\[
  I = \vr(C_i) = U^{-1}U\vr(C_i)U^{-1}U = U^{-1}\vr(UC_iU^{-1})U = U^{-1}\left(\begin{array}{ccc}\vr(\mu_{i1})&\dots&0\\ \vdots&\ddots&\vdots\\ 0&\dots&\vr(\mu_{i\ell}) \end{array}\right)U,
\]
implying $\vr(\mu_{ij})=1$. Then we proceed as in the first part of the claim.
\end{proof}

\begin{lemma}\label{lem:polynomial-image}
Let $\vr$ be a polynomial, $(\Omega,\cM,\mu)$ a measure space, and $c\in L^\infty(\Omega,\mu)$ with a finite range. Then $\vr(c)$ is bounded and measurable.
\end{lemma}

\begin{proof}
As $c$ has a finite range, so does $\rho(c)$. Therefore $\rho(c)$ is bounded. To
see that $c$ is measurable we use an argument similar to that from the proof
of Lemma~\ref{lem:lift}. For the event $Q_\sigma = \{\omega\in\Omega\mid
  \vr(c)(\omega) =\sigma\}$, for a fixed $\sigma\in U_d$,
we have $Q_\sigma = \{\omega\in\Omega\mid c(\omega) =\vr^{-1}(\sigma)\}$ and $Q_\sigma=\emptyset$ otherwise. As $c$ is a measurable function and $Q_\tau = c^{-1}(B_{1/2d}(\vr^{-1}(\tau)))$, where $B_{1/2d}(\tau)$ denotes the complex open ball of radius $1/2d$ centered
at $\tau$, which is a Borel set in the standard topology of $\zC$, $Q_\tau$ is a measurable set in the measure space $(\Omega,\cM,\mu)$. Since the range of $\vr(c)$ is finite, the preimage $(\vr(c))^{-1}(S)$ of each Borel subset $S$ of $\zC$ is expressed as a finite Boolean combination of measurable sets, and is thus measurable in $(\Omega,\cM,\mu)$.
\end{proof}

\begin{lemma}[Infinite-dimensional case]\label{lem:mapping-infinite}
Let $d,e\in\nat$, $e\le d$. 
Let $\pi:U_e\to U_d$ be an injective mapping and $\vr$ a unary polynomial that interpolates $B=\Im(\pi)$.  Let $\vc Ak$ be pairwise commuting bounded normal operators of order $e$ on a Hilbert space. Then $C_i=\pi(A_i)$, $i\in[k]$, are pairwise commuting bounded normal operators of order $d$, and $\vr(C_i)=I$. Conversely, let $\vc Ck$ be pairwise commuting bounded normal operators of order $d$ such that $\vr(C_i)=I$. Then for $A_i=\pi^{-1}(C_i)$, $i\in[k]$, it holds that the $A_i$'s are pairwise commuting bounded normal operators of order $e$, and there exist a measure space $(\Omega,\cM,\mu)$, a unitary map $U:\cH\to L^2(\Omega,\mu)$ and functions $c_1,\dots,c_k\in L^\infty(\Omega,\mu)$ such that $C_i = U^{-1} T_{c_i}U$ for $i\in [k]$ and the multiplication operators $T_{c_i}$ of $L^2(\Omega,\mu)$, and $\mu(\{\omega\in\Omega\mid c_i(\omega)\not\in B\}) = 0$ for $i\in[k]$.
\end{lemma}

\begin{proof}
By the General Strong Spectral Theorem (cf.~Theorem~\ref{the:general-sst}) there exist a measure space $(\Omega,\cM,\mu)$, a unitary map $U:\cH\to L^2(\Omega,\mu)$ and functions $a_1,\dots,a_k\in L^\infty(\Omega,\mu)$ such
that, for the multiplication operators $E_i = T_{a_i}$ of $L^2(\Omega,\mu)$, the relations $A_i = U^{-1} E_iU$ hold for each $i\in [k]$. Equivalently,
$U A_iU^{-1} = E_i$. From $A_i^e=I$ we conclude $E_i^e=I$. Hence, $a_i(\omega)^e = 1$ for almost all $\omega\in\Omega$; i.e., formally $\mu(\{\omega\in\Omega\mid a_i(\omega)^e\ne1\}) = 0$. Thus, $a_i(\omega)\in U_e$ for almost all $\omega\in\Omega$. Therefore, replacing $a_i$ with a different representative of its equivalence class in $L^\infty(\Omega,\mu)$ we may assume that $a_i(\omega)\in U_e$ for all $\omega\in\Omega$. Then we have 
\[
C_i=\pi(A_i)=\pi(U^{-1}E_iU)=U^{-1}\pi(E_i)U=U^{-1}\pi(T_{a_i})U=U^{-1}T_{\pi(a_i)}U.
\]
That the $C_i$'s are normal and pairwise commute follow from Lemma~\ref{lem:matrix-polys}, and it suffices to prove that $C^d_i=I$ and $c_i=\pi(a_i)$ are bounded and measurable.

The equality $C^d_i=I$ follows from the fact that $c_i(\omega)\in U_d$ for all $\omega\in\Omega$. By Lemma~\ref{lem:polynomial-image} $c_i$ is bounded and measurable. 

For the second part of the lemma, again, by the General Strong Spectral Theorem (cf. Theorem~\ref{the:general-sst}) there exist a measure space $(\Omega,\cM,\mu)$, a unitary map $U:\cH\to L^2(\Omega,\mu)$ and functions $c_1,\dots,c_k\in L^\infty(\Omega,\mu)$ such
that, for the multiplication operators $D_i = T_{c_i}$ of $L^2(\Omega,\mu)$, the relations $C_i = U^{-1} D_iU$ hold for each $i\in [k]$. As $\vr(C_i)=1$, we have 
\[
1=\vr(C_i)=\vr(U^{-1}D_iU)=U^{-1}\vr(D_i)U=U^{-1}\vr(T_{d_i})U=U^{-1}T_{\vr(d_i)}U.
\]
This implies 
\[
\mu(\{\omega\in\Omega\mid \vr(d_i(\omega))\ne0\})=\mu(\{\omega\in\Omega\mid d_i(\omega)\not\in B\}) = 0.
\]
The rest of the proof is identical to that of the first part of the lemma.
\end{proof}

Next we consider the four connections from Figure~\ref{fig:reductions} one by
one and prove that each of them preserves the satisfiability gap. For Steps
2--4, which rely on pp-definitions, a gap of the third kind is only preserved up
to the addition of $\rel_T$ to the language, as in the statement of
Theorem~\ref{the:hsp-gap}.

\smallskip

{\bf Step 1 (Reductions to a core).}
Let $\Gm$ be a constraint language over the set $U_d$ and $\vr:U_d\to U_d$ an idempotent
endomorphism of $\Gm$. Then $\CSP(\Gm)$ has a satisfiability gap if and only if $\CSP(\vr(\Gm))$ does.

\smallskip

Suppose that $|\Im(\vr)|=e$, let $\pi':\Im(\vr)\to U_e$ be any bijection between $U_e$ and $\Im(\vr)$, and $\pi=\pi'\circ\vr$. Let $\Dl=\{\pi(\rel)\mid \rel\in\Gm\}$, we show that $\CSP(\Dl)$ has a satisfiability gap if and only if $\CSP(\Gm)$ does.

Let $\cP=(V,U_e,\cC)$ be a gap instance of $\CSP(\Dl)$, and let $\cP^\pi=(V,U_d,\cC^\pi)$ be the corresponding instance of $\CSP(\Gm)$, where for each $\ang{\bs,\rel}\in\cC$ the set $\cC^\pi$ includes $\ang{\bs,\relo}$ with $\relo\in\Gm$ and $\pi(\relo)=\rel$. As is easily seen, $\cP^\pi$  has no solution, because for any solution $\vf$ of $\cP^\pi$ the mapping $\pi\circ\vf$ is a solution of $\cP$. 

{\it Finite-dimensional case.}
Let $\{A_v\mid v\in V\}$ be an $\ell$-dimensional satisfying operator assignment for $\cP$. We prove that $\{\pi'^{-1}(A_v)\mid v\in V\}$ is a satisfying operator assignment for $\cP^\pi$. Let $C_v=\pi'^{-1}(A_v)$.

By Lemma~\ref{lem:mapping}, the $C_v$'s are normal, $C_v^d=I$, $v\in V$,
and for any constraint $\ang{\bs,\rel}\in\cC$ and any $v,w\in\bs$, $C_v,C_w$ commute.
Now, let $\ang{\bs,\rel}\in\cC$, $\bs=(\vc vk)$, $\ang{\bs,\relo}$ be the
corresponding constraint of $\cP^\pi$, and $f_\rel,f_\relo$ be the polynomials representing $\rel,\relo$ over $U_e,U_d$ respectively. We have $f_\rel(A_{v_1}\zd A_{v_k})=I$, and we need to show that $f_\relo(C_{v_1}\zd C_{v_k})=I$. Let $U$ diagonalize $A_{v_1}\zd A_{v_k}$ and
\[
UA_{v_i}U^{-1}=\left(\begin{array}{ccc}\mu_{i1}&\dots&0\\ &\ddots&\\ 0&\dots&\mu_{i\ell}\end{array}\right).
\]
Then $(\mu_{1j}\zd\mu_{kj})\in\rel$ for $j\in[\ell]$, because
\[
I=f_\rel(A_{v_1}\zd A_{v_k})=f_\rel(UA_{v_1}U^{-1}\zd UA_{v_k}U^{-1})= \left(\begin{array}{ccc}f_\rel(\mu_{11}\zd\mu_{k1})&\dots&0\\ &\ddots&\\ 0&\dots&f_\rel(\mu_{1\ell}\zd\mu_{k\ell})\end{array}\right).
\]
Therefore, $(\pi'^{-1}(\mu_{1j})\zd\pi'^{-1}(\mu_{kj}))\in\vr(\relo)$ for $j\in[\ell]$. As $\vr$ is an endomorphism, $(\pi'^{-1}(\mu_{1j})\zd\pi'^{-1}(\mu_{kj}))\in\relo$ and so $f_\relo(\pi'^{-1}(\mu_{1j})\zd\pi'^{-1}(\mu_{kj}))=\ld_0=1$. Since 
\[
UC_{v_i}U^{-1}=U\pi'^{-1}(A_{v_i})U^{-1}=\pi'^{-1}(UA_{v_i}U^{-1})=\left(\begin{array}{ccc}\pi'^{-1}(\mu_{i1})&\dots&0\\ &\ddots&\\ 0&\dots&\pi'^{-1}(\mu_{i\ell})\end{array}\right),
\]
we also have $f_\relo(C_{v_1}\zd C_{v_k})=\ld_0I=I$.

{\it Infinite-dimensional case.}
Let $\{A_v\mid v\in V\}$ be an (infinite-dimensional) satisfying operator assignment for $\cP$. We prove that $\{\pi'^{-1}(A_v)\mid v\in V\}$ is a satisfying operator assignment for $\cP^\pi$. Let $C_v=\pi'^{-1}(A_v)$.

By Lemma~\ref{lem:mapping-infinite}, the $C_v$'s are normal, $C_v^d=I$, $v\in V$,
and for any constraint $\ang{\bs,\rel}\in\cC$ and any $v,w\in\bs$, $C_v,C_w$ commute.
Now, let $\ang{\bs,\rel}\in\cC$, $\bs=(\vc vk)$, $\ang{\bs,\relo}$ be the
corresponding constraint of $\cP^\pi$, and $f_\rel,f_\relo$ be the polynomials representing $\rel,\relo$ over $U_e,U_d$ respectively. We have $f_\rel(A_{v_1}\zd A_{v_k})=I$, and we need to show that $f_\relo(C_{v_1}\zd C_{v_k})=I$. By the General Strong Spectral Theorem (cf. Theorem~\ref{the:general-sst}) there exist a measure space $(\Omega,\cM,\mu)$, a unitary map $U:\cH\to L^2(\Omega,\mu)$ and functions $a_1,\dots,a_k\in L^\infty(\Omega,\mu)$ such that, for the multiplication operators $E_i = T_{a_i}$ of $L^2(\Omega,\mu)$, the relations $A_{v_i} = U^{-1} E_iU$ hold for each $i\in [k]$. As $A_{v_i}^e=I$, it holds that $a_i(\omega)\in U_e$ for almost all $\omega\in\Omega$. Choosing a different representative of the equivalence class of $a_i$ we may assume that $a_i(\omega)\in U_e$ for all $\omega\in\Omega$. Then $(a_1(\omega)\zd a_k(\omega))\in\rel$ for almost all $\omega\in\Omega$, because
\[
I=Uf_\rel(A_{v_1}\zd A_{v_k})U^{-1}=f_\rel(UA_{v_1}U^{-1}\zd UA_{v_k}U^{-1})= f_\rel(T_{a_i}\zd T_{a_k})=T_{f_\rel(a_i\zd a_k)},
\]
implying $f_\rel(a_i\zd a_k)(\omega)=1$ for almost all $\omega\in\Omega$. Therefore, $(\pi'^{-1}(a_1(\omega))\zd \pi'^{-1}(a_k(\omega)))\in\vr(\relo)$ for almost all $\omega\in\Omega$, and, as $\vr$ is an endomorphism, $(\pi'^{-1}(a_1(\omega))\zd \pi'^{-1}(a_k(\omega)))\in\relo$ and so $f_\relo(\pi'^{-1}(a_1(\omega))\zd \pi'^{-1}(a_k(\omega)))=\ld_0=1$ for almost all $\omega\in\Omega$. Since 
\[
UC_{v_i}U^{-1}=U\pi'^{-1}(A_{v_i})U^{-1}=\pi'^{-1}(UA_{v_i}U^{-1})=\pi'^{-1}(T_{a_i})=T_{\pi'^{-1}(a_i)},
\]
we also have $f_\relo(C_{v_1}\zd C_{v_k})=\ld_0I=I$. Moreover, $\pi'^{-1}(a_i)\in L^\infty(\Omega,\mu)$ by Lemma~\ref{lem:polynomial-image}.

\smallskip

Now, let $\cP^\pi=(V,U_d,\cC^\pi)$ be a gap instance of $\CSP(\Gm)$ and $\cP=(V,U_e,\cC)$, $\cC=\{\ang{\bs,\rel}\mid\ang{\bs,\relo}\in\cC^\pi,\rel=\pi(\relo)\}$, the corresponding instance of $\CSP(\Dl)$. Then again $\cP$ has no solution over $U_e$.

{\it Finite-dimensional case.}
Let $\{C_v\mid v\in V\}$ be an $\ell$-dimensional satisfying operator assignment for $\cP^\pi$. We need to prove that $\{A_v\mid v\in V\}$, $A_v=\pi(C_v)$ is a satisfying operator assignment for $\cP=(V,\cC)$. For $\ang{\bs,\rel}\in\cC$, $\bs=(\vc vk)$, let $f_\rel,f_\relo$ be polynomials representing $\rel$ and $\relo\in\Gm$ with $\pi(\relo)=\rel$, respectively. 

First, we show that the $C_v$'s can be replaced with $\vr(C_v)$. Since for any
$\ang{\bs,\relo}\in\cC^\pi$, $\bs=(\vc vk)$, and any $\vc ak\in U_d$, we have
that $f_\relo(\vr(a_1)\zd\vr(a_k))=1$ whenever $f_\relo(\vc ak)=1$, by Lemma~\ref{lem:lemma-3} $f_\relo(\vr(C_{v_1})\zd\vr(C_{v_k}))=I$ whenever $f_\relo(C_{v_1}\zd C_{v_k})=I$. Thus, we assume $A_v=\pi'(C_v)$ for $v\in V$.

That $A_v$ is normal, $A_v^e=I$, $v\in V$, and the $A_v$'s locally commute follows from Lemma~\ref{lem:mapping}. Let $f^\pi_\relo(\vc xk)=f_\relo(\pi'^{-1}(x_1)\zd\pi'^{-1}(x_k))$. As is easily seen, for any $\vc ak\in U_e$, if $f^\pi_\relo(\vc ak)=1$ then $f_\rel(\vc ak)=1$. Hence, by Lemma~\ref{lem:lemma-3} $f_\rel(A_{v_1}\zd A_{v_k})=I$ whenever $f^\pi_\relo(A_{v_1}\zd A_{v_k})=I$. Finally, we prove that $\pi'^{-1}(A_v)=C_v$; this implies that $f^\pi_\relo(A_{v_1}\zd A_{v_k})=I$ completing the proof. Let $U$ diagonalize $C_v$ and 
\[
UC_vU^{-1}=\left(\begin{array}{ccc}\mu_1&\dots&0\\ &\ddots&\\ 0&\dots&\mu_\ell\end{array}\right).
\]
Then
\begin{eqnarray*}
\pi'^{-1}(A_v) &=& \pi'^{-1}(\pi'(C_v))=U^{-1}U\pi'^{-1}(\pi'(C_v))U^{-1}U=U^{-1}\pi'^{-1}(\pi'(UC_vU^{-1}))U\\
&=& U^{-1}\left(\begin{array}{ccc}\pi'^{-1}(\pi'(\mu_1))&\dots&0\\ &\ddots&\\ 0&\dots&\pi'^{-1}(\pi'(\mu_\ell))\end{array}\right)U
= U^{-1}\left(\begin{array}{ccc}\mu_1&\dots&0\\ &\ddots&\\ 0&\dots&\mu_\ell\end{array}\right)U=C_v.
\end{eqnarray*}

{\it Infinite-dimensional case.}
Let $\{C_v\mid v\in V\}$ be a (infinite-dimensional) satisfying operator assignment for $\cP^\pi$. We need to prove that $\{A_v\mid v\in V\}$, $A_v=\pi(C_v)$ is a satisfying operator assignment for $\cP=(V,\cC)$. For $\ang{\bs,\rel}\in\cC$, $\bs=(\vc vk)$, let $f_\rel,f_\relo$ be polynomials representing $\rel$ and $\relo\in\Gm$ with $\pi(\relo)=\rel$, respectively. As in the finite-dimensional case it can be shown that the $C_v$'s can be replaced with $\vr(C_v)$, and therefore we assume $A_v=\pi'(C_v)$ for $v\in V$.

That $A_v$ is normal, $A_v^e=I$, $v\in V$, and the $A_v$'s locally commute follows from Lemma~\ref{lem:mapping-infinite}. As before, for any $\vc ak\in U_e$, if $f^\pi_\relo(\vc ak)=1$ then $f_\rel(\vc ak)=1$, implying by Lemma~\ref{lem:lemma-3} that $f_\rel(A_{v_1}\zd A_{v_k})=I$ whenever $f^\pi_\relo(A_{v_1}\zd A_{v_k})=I$. Finally, we prove that $\pi'^{-1}(A_v)=C_v$; this implies that $f^\pi_\relo(A_{v_1}\zd A_{v_k})=I$ completing the proof. By the General Strong Spectral Theorem (cf. Theorem~\ref{the:general-sst}) there exist a measure space $(\Omega,\cM,\mu)$, a unitary map $U:\cH\to L^2(\Omega,\mu)$ and functions $c_1,\dots,c_k\in L^\infty(\Omega,\mu)$ such that, for the multiplication operators $D_i = T_{c_i}$ of $L^2(\Omega,\mu)$, the relations $C_{v_i} = U^{-1} D_iU$ hold for each $i\in [k]$. As before, the $c_i$'s can be chosen such that $c_i(\omega)\in U_d$ for all $\omega\in\Omega$. Then by Lemma~\ref{lem:polynomial-image} $\pi'^{-1}(\pi'(c_i))\in L^\infty(\Omega,\mu)$ and 
\begin{align*}
\pi'^{-1}(A_v) &= \pi'^{-1}(\pi'(C_v))=U^{-1}U\pi'^{-1}(\pi'(C_v))U^{-1}U\\
&=U^{-1}\pi'^{-1}(\pi'(UC_vU^{-1}))U=U^{-1}\pi'^{-1}(\pi'(T_{c_i}))U=U^{-1}T_{\pi'^{-1}(\pi'(c_i))}U.
\end{align*}
Since $\pi'^{-1}(\pi'(c_i(\omega)))=c_i(\omega)$ for almost all $\omega\in\Omega$, we obtain $U^{-1}T_{\pi'^{-1}(\pi'(c_i))}U=C_{v_i}$.

\medskip

{\bf Step 2  (Adding constant relations).}
Let $\Gm$ be a core language. Then $\CSP(\Gm)$ has a satisfiability gap if and only if $\CSP(\Gm^*)$ does.

\smallskip

Since $\Gm\sse\Gm^*$, if $\CSP(\Gm)$ has a satisfiability gap, so does
$\CSP(\Gm^*)$. We prove that if $\CSP(\Gm^*)$ has a satisfiability gap then
$\CSP(\Gm)$ has a satisfiability gap.

We will use the following relation $\rel_\Gm$ that is known to be pp-definable in $\Gm$ \cite{Jeavons99:expressive}: 
\[
\rel_\Gm=\{(\vr(\ld_0)\zd\vr(\ld_{d-1}))\mid \text{$\vr$ is an endomorphism of $\Gm$}\}.
\]
As $\Gm$ is a core language, for any $(\vc ad)\in\rel_\Gm$ it holds that $\{\vc
ad\}=U_d$. By Theorem~\ref{the:pp} we may assume that $\rel_\Gm\in\Gm$.

Let $\cP=(V,U_d,\cC)$ be a gap instance of $\CSP(\Gm^*)$. We construct an instance $\cP'=(V',U_d,\cC')$ of $\CSP(\Gm)$ as follows. 
\begin{itemize}
    \item 
    $V'= V \cup \{v_a | a \in U_d \}$;
    \item
    $\cC'$ consists of three parts: $\{C=\ang{\bs,\rel}\in\cC\mid \rel\in\Gm\}$, $\{\ang{(v_{a_1}\zd v_{a_n}),\rel_\Gm}\}$, and  $\{\ang{(v,v_a),=_d}\mid \ang{(v),C_a}\in\cC\}$, where $=_d$ denotes the equality relation on $U_d$.
\end{itemize}

It is known \cite{Jeavons99:expressive} that $\cP'$ has a classic solution if and only if $\cP$ has one. If $\vf:V\to U_d$ is a solution of $\cP$ then we can extend it to a solution of $\cP'$ by mapping $v_a$ to $a$. Conversely, let $\vf:V'\to U_d$ be a solution of $\cP'$. The restriction of $\vf$ to $V$ may not be a solution of $\cP$, because for some constraint $\ang{(v),C_a}\in\cC$ it may be the case that $\vf(v)=\vf(v_a)\ne a$. This can be fixed as follows.  As $\vf$ is a solution, $(\vf(v_{\ld_0})\zd \vf(v_{\ld_{d-1}}))\in\rel_\Gm$, hence, the mapping $\vr:U_d\to U_d$ given by $\vr(a)=\vf(v_a)$ is an endomorphism of $\Gm$.  As $\Gm$ is a core language, $\vr$ is a permutation on $U_d$ and $\vr^s$ is the identity permutation for some $s$. Then $\vr^{s-1}$ is the inverse $\vr^{-1}$ of $\vr$ and is also an endomorphism of $\Gm$. Therefore $\vf'=\vr^{-1}\circ\vf$ is also a solution of $\cP'$ and $\vf'(v_a)=a$ for $a\in U_d$. Thus, $\vf'\red V$ is a solution of $\cP$.

\smallskip

{\it Finite-dimensional case.}
Suppose that $\{A_v\mid v\in V\}$ is an $\ell$-dimensional satisfying operator assignment for $\cP$. First, we observe that if $\cC$ contains a constraint $\ang{(v),C_a}$ then $A_v$ is the scalar operator $aI$. Indeed, let $f_a(x)$ be a polynomial representing $C_a$, that is,  $f_a(a)=\ld_0$ and $f_a(b)=\ld_1$ for $b\in U_d-\{a\}$. Let also $U$ be a unitary operator that diagonalizes $A_v$ and $\vc\mu\ell$ the eigenvalues of $A_v$. Then, as $f_a(A_v)=I$ we obtain
\[
I=UIU^{-1}=Uf_a(A_v)U^{-1}=f_a(UA_vU^{-1})=
\left(\begin{array}{ccc}f_a(\mu_1))&\dots&0\\ &\ddots&\\ 0&\dots&f_a(\mu_\ell))\end{array}\right)
\]
implying that $f_a(\mu_i)=\ld_0$ for $i\in[\ell]$. Thus, $a$ is the only eigenvalue of $A_v$ and 
\[
A_v=U^{-1}aIU=aI.
\]
All such operators pairwise commute regardless of the value of $a$. Therefore $v_a$ can be assigned $aI$ for $a\in U_d$, and the resulting assignment is a satisfying operator assignment for $\cP'$.

\smallskip

{\it Infinite-dimensional case.}
Suppose that $\{A_v\mid v\in V\}$ is a (infinite-dimensional) satisfying operator assignment for $\cP$. As before, we observe that if $\cC$ contains a constraint $\ang{(v),C_a}$ then $A_v$ is the scalar operator $aI$. Let $f_a(x)$ be a polynomial representing $C_a$, that is,  $f_a(a)=\ld_0$ and $f_a(b)=\ld_1$ for $b\in U_d-\{a\}$. By the General Strong Spectral Theorem (cf. Theorem~\ref{the:general-sst}) there exist a measure space $(\Omega,\cM,\mu)$, a unitary map $U:\cH\to L^2(\Omega,\mu)$ and a function $c\in L^\infty(\Omega,\mu)$ such that, for the multiplication operator $E = T_c$ of $L^2(\Omega,\mu)$, the relation $A_v = U^{-1} EU$ holds. By Lemma~\ref{lem:lemma-3} 
\[
f_a(E)=f_a(UA_vU^{-1})=Uf_a(A_v)U^{-1}=UIU^{-1},
\]
implying $f_a(c(\omega))=1$ for almost all $\omega\in\Omega$. Therefore, $A_v=U^{-1}EU=U^{-1}aIU=aI$. We then complete the proof as in the finite-dimensional case.

\medskip

{\bf Step 3 (Satisfiability gap from subalgebras).}
Let $\Gm$ be a constraint language over the set $U_d$ and let $B$ be its subalgebra.
Then if $\CSP(\Gm\red B)$ has a satisfiability gap then so does $\CSP(\Gm)$.

\smallskip

Let $\Dl=\Gm\red B$. Then by Theorem~\ref{the:pp} we may assume $\Dl\sse\Gm$ and $B\in\Gm$. Let $e=|B|$ and $\pi:U_e\to U_d$ a bijection between $U_e$ and $B$. 

Let $\cP=(V,U_e,\cC)$ be a gap instance of $\CSP(\pi^{-1}(\Dl))$ and the instance $\cP^\pi=(V,U_d,\cC^\pi)$ constructed as follows. For every $\ang{\bs,\rel}\in\cC$ the instance $\cP^\pi$ contains $\ang{\bs,\pi(\rel)}$. As is easily seen, $\cP^\pi$ has no classic solution. Therefore, it suffices to show that for any finite- or infinite-dimensional satisfying operator assignment $\{A_v\mid v\in V\}$ for $\cP$, the assignment $C_v=\pi(A_v)$ is a satisfying operator assignment for $\cP^\pi$.

By Lemmas~\ref{lem:mapping},~\ref{lem:mapping-infinite}, the $C_v$'s are normal, satisfy the condition $C_v^d=I$, locally commute, and bounded in the infinite-dimensional case. For $\ang{\bs,\rel}\in\cC$, $\bs=(\vc vk)$, let $f^\pi_\rel(\vc xk)=f_\rel(\pi^{-1}(x_1)\zd\pi^{-1}(x_k))$. As in \textbf{Step 1}, it can be shown that $\pi^{-1}(C_v)=A_v$, and therefore $f^\pi_\rel(C_{v_1}\zd C_{v_k})=I$. For any $\vc ak\in U_d$, if $(\vc ak)\in\pi(\rel)$ then $\vc ak\in B$. Therefore, if $f^\pi_\rel(\vc ak)=1$ then $f_{\pi(\rel)}(\vc ak)=1$. By Lemma~\ref{lem:lemma-3} this implies $f_{\pi(\rel)}(C_{v_1}\zd C_{v_k})=I$.

\smallskip

{\bf Step 4 (Satisfiability gap from homomorphic images).}
Let $\Gm$ be a constraint language over the set $U_d$ and $\th$ a congruence of
$\Gm$.  Then if $\CSP(\Gm\fac\th)$ has a satisfiability gap then so does $\CSP(\Gm)$.

\smallskip

Let $\vr:U_d\to U_d\fac\th$ be the natural mapping $a\mapsto a\fac\th$ and $\pi':U_d\fac\th\to U_e$, where $e=|U_d\fac\th|$, a bijection. Finally, let $\pi=\pi'\circ\vr$ and $\Dl=\pi(\Gm)$. For $\rel\in\Dl$ let $\pi^{-1}(\rel)$ be the full preimage of $\rel$ under $\pi$. Since $\th$ is pp-definable in $\Gm$, so is $\pi^{-1}(\rel)$ for any $\rel\in\Dl$. Indeed, if $\rel=\pi(\relo)$ for some $\relo\in\Gm$, then 
\[
\pi^{-1}(\rel)(\vc xk)=\exists\vc yk\ \ \relo(\vc yk)\wedge\bigwedge_{i\in[k]}\th(x_i,y_i).
\]
Using Theorem~\ref{the:pp} we may assume that $\pi^{-1}(\rel)\in\Gm$ for $\rel\in\Dl$. Let $\pi^*:U_e\to U_d$ assign to $a\in U_e$ a representative of the $\th$-class $\pi'^{-1}(a)$. Thus, in a certain sense, $\pi^*$ is an inverse of $\pi$.  

Suppose that $\cP=(V,U_e,\cC)$ is a gap instance of $\CSP(\Dl)$ and let $\cP^\pi=(V,U_d,\cC^\pi)$ be given by $\cC^\pi=\{\ang{\bs,\pi^{-1}(\rel)}\mid \ang{\bs,\rel}\in\cC\}$. We prove that $\cP^\pi$ is a gap instance of $\CSP(\Gm)$. Firstly, observe that $\cP^\pi$ has no classic solution, because for any solution $\vf$ of $\cP^\pi$ the mapping $\pi\circ\vf$ is a solution of $\cP$. Let $\{A_v\mid v\in V\}$ be a satisfying operator assignment for $\cP$. We set $C_v=\pi^*(A_v)$ and prove that $\{C_v\mid v\in V\}$ is a satisfying operator assignment for $\cP^\pi$. By Lemma~\ref{lem:mapping},~\ref{lem:mapping-infinite}, the $C_v$'s are normal, satisfy the condition $C_v^d=I$, bounded, and locally commute. For $\ang{\bs,\rel}\in\cC$, $\bs=(\vc vk)$, let $f^\pi_\rel(\vc xk)=f_\rel(\pi(x_1)\zd\pi(x_k))$. As before, it is easy to see that $\pi(C_v)=A_v$ for $v\in V$, and therefore $f^\pi_\rel(C_{v_1}\zd C_{v_k})=I$. Finally, for any $(\vc ak)\in U_d$, if $f^\pi_\rel(\vc ak)=f_\rel(\pi(a_1)\zd\pi(a_k))=1$, then $(\pi(a_1)\zd\pi(a_k))\in\rel$, and so $(\vc ak)\in\pi^{-1}(\rel)$ and $f_{\pi^{-1}(\rel)}(\vc ak)=1$. By Lemma~\ref{lem:lemma-3} this implies that $f_{\pi^{-1}(\rel)}(C_{v_1}\zd C_{v_k})=I$.

\section*{Acknowledgements} We are grateful to William Slofstra for useful
discussions, for telling us about several relevant results, and finally for finding mistakes in a previous version of this paper.

{\small
\bibliographystyle{plainurl}
\bibliography{bz}
}

\end{document}